\documentclass{article}
\pdfoutput=1

\usepackage{arxiv}

\usepackage[utf8]{inputenc} 
\usepackage[T1]{fontenc}    
\usepackage{hyperref}       
\usepackage{url}            
\usepackage{booktabs}       
\usepackage{amsfonts}       
\usepackage{nicefrac}       
\usepackage{microtype}      
\usepackage{lipsum}		
\usepackage{graphicx}
\usepackage{natbib}
\usepackage{doi}

\setcitestyle{authoryear,open={(},close={)}}
\usepackage{bm}
\newcommand{\dd}{\mathrm{d}}
\newcommand\numberthis{\addtocounter{equation}{1}\tag{\theequation}}
\usepackage{physics}
\usepackage{tikz}
\usepackage{program}
\usepackage{mathdots}
\usepackage{yhmath}
\usepackage{cancel}
\usepackage{color}
\usepackage{siunitx}
\usepackage{array}
\usepackage{multirow}
\usepackage{gensymb}
\usepackage{tabularx}
\usepackage{extarrows}
\usepackage{booktabs}
\usetikzlibrary{fadings}
\usetikzlibrary{patterns}
\usetikzlibrary{shadows.blur}
\usetikzlibrary{shapes}
\usepackage{placeins}
\usepackage{algorithm}
\usepackage[noend]{algpseudocode}
\usepackage{appendix}
\usepackage{caption}
\usepackage{subcaption}

\newtheorem{proposition}{Proposition}
\newtheorem{assumption}{Assumption}

\title{The Node-wise Pseudo-marginal Method}

\author{
    Denishrouf Thesingarajah\\
	Department of Statistics\\
	University of Warwick\\
	Coventry, CV4 7AL\\
	\texttt{d.thesingarajah@warwick.ac.uk} \\
	\And
    Adam M.~Johansen\\
	Department of Statistics\\
	University of Warwick\\
	Coventry, CV4 7AL\\
	\texttt{a.m.johansen@warwick.ac.uk} \\
}

\renewcommand{\headeright}{}
\renewcommand{\undertitle}{}
\renewcommand{\shorttitle}{The NWPM Method}

\hypersetup{
pdftitle={The Node-wise Pseudo-marginal Method},
pdfsubject={q-bio.NC, q-bio.QM},
pdfauthor={Denishrouf Thesingarajah, Adam M.~Johansen},
pdfkeywords={Bayesian inference, Monte Carlo, poistron emission tomography, spatial modelling},
}

\begin{document}
\maketitle

\begin{abstract}
Motivated by problems from neuroimaging in which existing approaches make use of ``mass univariate'' analysis which neglects spatial structure entirely, but the full joint modelling of all quantities of interest is computationally infeasible, a novel method for incorporating spatial dependence within a (potentially large) family of model-selection problems is presented. Spatial dependence is encoded via a Markov random field model for which a variant of the pseudo-marginal Markov chain Monte Carlo algorithm is developed and extended by a further augmentation of the underlying state space. This approach allows the exploitation of existing unbiased marginal likelihood estimators used in settings in which spatial independence is normally assumed thereby facilitating the incorporation of spatial dependence using non-spatial estimates with minimal additional development effort.

The proposed algorithm can be realistically used for analysis of 
moderately sized data sets such as $2$D slices of whole $3$D dynamic PET brain images or other regions of interest.
Principled approximations of the proposed method, together with simple extensions based on the augmented spaces, are investigated and shown to provide similar results to the full pseudo-marginal method.
Such approximations and extensions allow the improved performance obtained by incorporating spatial dependence to be obtained at negligible additional cost.
An application to measured PET image data shows notable improvements in revealing underlying spatial structure when compared to current methods that assume spatial independence.\end{abstract}


\section{Introduction}\label{sec1}

Incorporating spatial dependence in the process of analysing large image, and other spatial image--like, data sets can be a difficult problem largely due to the computational requirements.
An important example of such data is provided by PET (Positron Emission Tomography) imaging of the brain, where a whole image typically requires analysis of up to $10^{6}$ time series \citep{Hammers2007}.
Current state of the art analysis, e.g. \citet{Fan2021,Zhou2016}, of such data generally either assumes spatial independence of pixels or performs large-scale aggregation over space \citep{Zhou2002} to overcome computational restrictions. In this paper, a novel computational method to incorporate spatial dependence in the process of model selection is presented.

The technique presented here extends the well-known pseudo-marginal Markov Chain Monte Carlo (MCMC) algorithm originally presented by \cite{Beaumont2003}, and characterised by \cite{Andrieu2009}.
More specifically, a Potts model \citep{Potts1952} is first utilised as a prior distribution to encode spatial dependence over a \emph{local} model associated with each node of a graph.
Imposing a further assumption of conditional spatial independence for local model parameters given the local model then gives rise to significant simplifications in computations.

This Markov random field model is used on the discrete space of model orders to describe spatial dependence between neighbouring nodes (pixels).
Finally, a standard component-wise MCMC updating scheme is used in conjunction with the above assumptions.
This then allows for the use of node-level marginal likelihood estimators to be used within a pseudo-marginal MCMC method.
The end result is a tractable, accessible computational method that uses model selection at the individual (node/pixel) level to perform spatially dependent model selection on the whole (graph/image) data set.
This method provides a flexible but efficient algorithm that can be readily implemented ---
in fact, the proposed method allows for the use of existing unbiased non-spatial estimators to be used within a framework which incorporates spatial dependence.

In addition, the proposed method can be further extended through augmentation of the state space of the MCMC chain; in this specialised setting careful specification of the proposal distributions can lead to many further computational approaches and techniques.

The remainder of this paper is structured as follows.
Section~\ref{section:Background} presents relevant background material including the motivating problem  and reviews the use of Potts models in related contexts. Section~\ref{sec:existing-monte-carlo} reviews existing methods for sampling from distributions of this type.
Section~\ref{section:Method} presents a hierarchical model and develops methods for using it in the context of interest alongside a formal justification of the approach and a number of extensions. Finally, the performance of the developed method, and some variants thereof, is studied empirically in applied settings in Section~\ref{section:EmStudy}.

\section{ Background}
\label{section:Background}
\subsection{Preliminaries}
A graph $G=(V,E)$ is the pair of sets of nodes, or vertices, denoted $V$ and edges denoted $E$. 
In particular, each element of the edge-set $E$ is some pair of elements of the nodes-set $u,v \in V$, denoted $\langle u,v \rangle \in E$.
Two nodes $u,v$ are said to be connected by the edge $\langle {u,v} \rangle$ if $\langle u,v \rangle \in E$.
In this case we may say that $u$ and $v$ are neighbours, or adjacent, and denote this by the relation $u \sim v$. 
Here, we will look only at undirected graphs, so the relation $\sim$ will be symmetric and $\langle {u,v}  \rangle$ is an \emph{unordered} pair.
Denote by $\partial(v) = \{u \in V : v \sim u \}$  the set of neighbours of $v$, by convention $v$ is not a neighbour of itself. 

Given a graph, $G,$ a collection of random variables, $\bm{X} = (X_v : v \in V)$, indexed by the nodes-set, $V$, is called a random field on $G$.
Let $P$ be the law of $\bm{X}$; and define, for any $A \subset V$, $\bm{X}_A = \{X_v : v \in A \}$ and $ \bm{X}_{-A} = \{X_v : v \in V \setminus A \}$. With a slight abuse of this notation, we will write $\bm{X}_{-v}$ to denote $\bm{X}_{V \setminus \{v\}}$ for $v \in V$.
$\bm{X}$ is called a Markov random field \citep{Besag1974} on a discrete graph if and only if we have that:
$$P(X_{v} \vert \bm{X}_{-v}) = P(X_{v}\vert\bm{X}_{\partial(v)}),$$
where each component $X_v$ takes value in some set $\mathcal{X}$.
Thus, $\bm{X}$ takes values in $\mathcal{X}^{V}$, the collection of maps from $V$ to $\mathcal{X}$.
In this paper, $\mathcal{X}$ will always be finite.

Denote a parametric model for data $Y \in \mathcal{Y}$, by \( S =(\mathcal{Y},\{ P_\theta: \theta \in \Theta \})\); where $P_\theta$ is a probability distribution over $\mathcal{Y}$, with $\theta$ taking values in the  parameter space $\Theta \subset \mathbb{R}^d$ with $d < \infty$.
We will use $f(\cdot \vert \theta)$ to denote the density (w.r.t some dominating measure, usually the Lebesgue or counting measures) of the distribution $P_{\theta}$.
Allow $Y \sim P_\theta $ and, where the density exists, $ Y \sim f(\cdot \vert \theta),$ to mean that $Y$ is distributed according to the distribution $P_\theta$. 

A finite set of statistical models, called the model space, is denoted $\mathcal{S} = \{ S_M : M \in \mathcal{M} \}$. Here, each model is indexed by the model order $M$,  which is an element of a finite index set $\mathcal{M} \subset \mathbb{N}$.
For a model $S_M$, of order $M$, let $f_M(\cdot \vert \theta)$ for parameter $\theta \in \Theta_M$, be the associated density.
We write $\Theta_M$ to emphasis that the parametric space is often dependent on the model order.
A straightforward example of such models are the compartmental models, see \citep{Gunn2001} or Section \ref{PET} below.
Bayesian model selection, averaging and comparison involves inference of the model order $M \in \mathcal{M}$, where $M$ is a random variable to which some prior distribution is attached \citep[Chapter 7]{Robert2007}. 
An important component of the Bayesian model selection process is the \emph{posterior model probability}, denoted $\pi (M\vert Y)$.
Specifically, a central objective here is to characterise this distribution in the setting for image and image-like data, while accounting for spatial dependence.
The definition of the posterior distribution (particularly for the setting we are interested in) will be detailed in the sequel.

\subsection{Position Emission Tomography}
\label{PET}
An application that motivates the methodological development that follows is
Positron Emission Tomography (PET; \citet{Phelps2006}). 

Dynamic PET is an increasingly important neuroimaging technique used to study
different functions of the brain \emph{in vivo}. PET images have been used to investigate
many biochemical processes in the brain, in both control and atypical subjects.
This invasive imaging modality uses the detection of high-energy photons,
released as a consequence of an injected tracer's positron emissions, as a
signal to form a dynamical three-dimensional image.

Given that there will be greater radioactivity in regions where there are higher amounts of the tracer, the signal intensity in each voxel will be proportional to the concentration of the tracer. Thus, by taking a sequence of measurements over time, the constructed image represents the time course of the tissue concentration of the tracer \emph{in vivo}. 

As a result of the nature and number of observations recorded in the
reconstructed PET image, the data is usually modelled using a system of linear
differential equations such as compartmental models \citep{Gunn2001}. We restrict our attention here to the plasma input compartment models. Given data $y = (y_1, \dots,y_k)^{\top}$, at
a voxel, an $M-$compartment model can be formally written as:
\begin{align}
  y_j = C_{\mathrm{T}}(t_j;\phi_{1:M}, \vartheta_{1:M}) + \sqrt{\frac{C_{
  \mathrm{T}}(t_j;\phi_{1:M}, \vartheta_{1:M})}{t_j - t_{j-1}}} \epsilon_j, \label{eq:comp_signal} \\
  C_{\mathrm{T}}(t_j;\phi_{1:M}, \vartheta_{1:M}) = \sum_{i=1}^{M} \phi_i \int_{0}^{t_j}C_{\mathrm{P}}(s)e^{-\vartheta_i(t_j-s)}ds,\label{eq:comp_model} 
\end{align}
for $j = 1, \dots, k$; where $C_{\mathrm{T}}$ is the tissue concentration, $t_j$ is the measurement time of
observation $y_j$, $\epsilon_j$ is the additive measurement error and the
plasma input function $C_{\mathrm{P}}$ is treated as known.
The parameters \(\phi_{1:M} = (\phi_1, \dots, \phi_M)\) and \( \vartheta_{1:M} = ( \vartheta_1, \dots, \vartheta_M )\)
determine the dynamics of the model.
Importantly, these parameters are
indistinguishable, \citep[Theorem 2.2]{Gunn2001}, and are often called the micro-parameters. In contrast, the macro-parameter
$$V_{\mathrm{D}} = \sum_{i=1}^{M}\frac{\phi_i}{\vartheta_i},$$
termed the volume of distribution is uniquely identifiable. This macro-parameter
is very often the quantity of principal interest in statistical analysis of PET
data. 

The additive error variable $\epsilon_j$ is typically modelled using a normal distribution:
$$\epsilon_j \sim \mathcal{N}(0,\sigma^2),$$
where $\mathcal{N}(0,\sigma^2)$ is the normal distribution with mean zero and variance $\sigma^2$.
This class of statistical models have proven to be very effective in Bayesian analysis of PET
data, as studied and presented by \citet{Peng2008} and \citet{Jiang2009}.

One ubiquitous method for PET analysis is the non-negative non-linear least squares (NNLS; \cite{Cunningham1993}).
This approach is often used due to the non-negative nature of the rate constants.
The NNLS method involves minimal modelling assumptions, imposing no prior assumptions on the number of components.
Instead, the time-series at each voxel is interpreted as a noisy measurement of exponential decay.
The method involves the formulation of a constrained linear optimisation problem, where: $M$ will represent the maximum number of terms to be included, for example Cunningham and Jones (1993) use $M = 100$;
Each component of the $\vartheta$ parameter is fixed and predetermined to lie within some range that is physiologically meaningful.
Next, the optimal value of $\phi_j$ for each $\vartheta_j$ is estimated using a numerical method, subject to $\phi_i \geq 0, j=1, \dots, m$.

Recently, computer-intensive Monte Carlo methods have also been successfully
applied to these models towards meaningful statistical inference of PET data.
Specifically, \cite{Zhou2013} investigated an
application specific MCMC approach to Bayesian model comparison for this class
of models; and studied both vague non--informative and biologically informed priors. In
this work, samples from a MCMC chain were used in the generalized harmonic mean estimator to
approximate the model evidence in order to perform model comparison, selection and
averaging. This work also demonstrated the possibility that Monte Carlo approaches could be used
to investigate
and meaningfully compare more complex models, containing up to $M=3$ compartments, rather
than the typical up to $M=2$ compartmental models.
\cite{Zhou2013} further showed that within a Bayesian framework, a $t-$distributed error structure may be far more plausible. This can be written as
$$\epsilon_j \sim \mathcal{T}(0,\tau, \nu),$$
where $\mathcal{T}(0,\tau,\nu)$ denotes the Student's $t-$distribution with mean zero, scale $\tau$, and $\nu$ degrees of freedom.
Expressions for prior and posterior densities, of both cases, can be found in Appendix \ref{App:Model_Equations}.

Importantly, with regard to model selection within this context, \citet{Zhou2013} also showed that a Bayesian approach \emph{exhibited some spatial structure} despite assuming spatial independence.
This was in contrast to using the Akaike Information Criteria (AIC; \cite{Akaike1973}) for model selection with a non-linear least square (NLS) method (as well as the NNLS method) used to approximate the MLE which showed no obvious spatial structure, see also Section \ref{sec:data-analys}.

More recently, \cite{Zhou2016} used an adaptive variant of the sequential Monte Carlo (SMC)
sampler algorithm of \cite{DelMoral2006} to estimate the model evidence (marginal likelihood). The class of algorithms
presented in this work, use adaptive MCMC kernels together with adaptive
annealing schemes to minimise tuning requirements to further the accessibility of
Monte Carlo approaches to this problem. In addition, the model evidence can be
directly approximated using the set of weights from the sampler. The
computational method proposed in this present paper, builds upon this approach
and thus inherits many of these advantages.
Further, \citet{Zhou2016} presented empirical results suggesting that this SMC method produces similar performance to the MCMC method proposed by \citet{Zhou2013};
As such, in this paper, we focus on comparing performance to the SMC sampler method --- which we refer to as the (spatially) independent SMC method.

Similarly, \cite{Castellaro2017} propose a method based on a Variational Bayes (VB) approach for parameter estimation in PET data analysis; and also present an empirical comparison to the MCMC method of \cite{Zhou2013}.
Specifically, in this empirical comparison the $M=2-$compartment model was used to represent measured PET data.
The VB method was adapted to kinetic modelling of PET data, applied to the $M=2-$compartment model and used to estimate the $V_{\rm{D}}$.
Next, the MCMC method, restricted to inference from the $M=2-$compartment model, was also used to analyse this data.
Castellaro et al. concluded that there was a strong agreement between the outputs of these two methods. In particular the $V_{\rm{D}}$ maps of both the VB and MCMC methods showed a very high correlation (Pearson $\rm{R}^2 = 0.99$ (see, Figure S2, \citet{Castellaro2017})).

There continues to be further development in Bayesian approaches in analysis of PET data, albeit in almost all cases they are non-spatial approaches. 
The most recent example is \citet{Fan2021}, who propose a simple and intuitive algorithm, based on Approximate Bayesian Computation(ABC), for analysis of PET data.
As with the above Bayesian approaches, this method, termed PET-ABC, assumes voxels are spatially independent. PET-ABC works by simulating from the parameter space using a simple rejection scheme: Firstly, proposals are sampled from the prior distribution, typically this is the uniform density with a physiologically meaningful range.
Each proposal is used as a trial value to estimate the signal $C_{\text{T}}$.
Next, the error between a summary statistics of the estimated signal and the de-noised data $\bm{y}$ is computed.
In their study,  \citet{Fan2021} suggest using spline smoothed estimates of $C_\text{T}$ as the summary statistic for both the estimated and observed signal.
The proposed value of the parameter is accepted if the above error is below a predetermined threshold.
The generated parameter sample can then be use for point estimation and to quantify uncertainty.
Bayesian model selection follows the above method, with the additional step of proposing a model at each iteration.  
\cite{Fan2021} state that using PET-ABC produced estimates with lower variance when compared to the weighted NLS method. 

\subsection{Potts Model}
\label{sec:potts-model}

The Potts model \citep{Potts1952}, a generalisation of the Ising model \citep{Ising1925}, was used originally to model interacting spins on a lattice.
However, these models have been shown to be also very effective in analysis of image data; for a detailed discussion and review of such applications see \citet{Winkler1995}, \citet{Geman1990} and references therein.
For example, \citet{Geman1984} is an early study demonstrating the effectiveness of the Ising model for restoration of images under a Bayesian framework.
More recently, the Potts model has been used very successfully within the more broader but still growing sub-field of Bayesian image analysis \cite{Hurn2003}.

In the present context, the Potts model is an ideal prior for model orders, due to its discrete state  space, minimal parametrisation and ability to encode the general principle that nearby vertices are \emph{a priori} likely to be best described by the same model.
Following the Bayesian approach, in this study, the model order of the data will be Markov random field with a Potts prior distribution.

For a review of Potts distributions, particularly in the image analysis context, see \cite{Hurn2003}.

Given a graph $G = (V,E)$, finite state space $\mathcal{X}$ and coupling constant $J>0$, the Potts model specifies a family of parametric probability distributions on $\mathcal{X}^{V}$; characterised by the  joint mass function:
$$p(\bm{x} \vert J,G) = \frac{1}{\zeta(J)}\exp\left( J \sum_{v \sim u} \delta_{x_v, x_u}\right).$$
Here, recall $v \sim u$ denotes neighbouring pairs and $\delta_{x_v, x_u}$ is the Kronecker delta, i.e. $\delta_{x_v, x_u}$ is one if $x_v = x_u$, and zero otherwise.
The intractable normalising constant (or partition function), written
$$\zeta(J)  = \sum_{\bm{x'} \in \mathcal{X}^{V}} \exp\left( J \sum_{v \sim u} \delta_{x_v', x_u'} \right),$$
is a function of $J$. 

The parameter $J$ dictates how likely neighbouring random variable $X_v$ and $X_u$ are to have the same value.
Given that the normalisation depends upon this parameter, \(J\) is  typically very difficult to infer; and will be treated as known in this setting as discussed at the end of this section.

Given some graph $G = (V,E)$, and the associated random variable $\bm{X} = (X_v \in \mathcal{X}: v \in V)$, we write
$$\bm{X} \sim \text{Potts}(J, G, \mathcal{X})$$
to mean that the random variable $\bm{X}$, taking values in state space $\mathcal{X}^{V}$, is distributed according to the Potts model with given coupling constant $J$.
Where no ambiguity arises, the parameters will be omitted from notation in the interests of clarity.

The Potts model exhibits phase transition behaviour.
This has significant implications when using single-site update scheme MCMC methods, which could result in slow mixing.
That is, if the chain is in a configuration such as the case described above, changes to single variables will mean proposing to move to states of lower probability. 
Furthermore, phase transition behaviour happens close to or higher than critical values of the coupling constant, which we denote as $J_\text{critical}$.
In the Monte Carlo setting, there is a sharp transition from fast mixing below $J_{\text{critical}}$ and slow mixing above it.
\citet{Onsager1944} showed that for the Ising model on the two dimensional first order square lattice the exact value was
 $$J_\text{critical} = \log\left( 1+\sqrt{2} \right) \approx 0.881.$$
 More recently, see for instance \citet{Matveev1996} or \citet{Wu1982}, this has been extended: Let $D= \vert \mathcal{X} \vert $, then for a
$D-$state Potts model,
$$J_\text{critical} = \log\left( 1 + \sqrt{D} \right).$$

In applications, since the intractable normalising constant $\zeta(J)$ is dependent on $J$, it is difficult to infer in the model fitting process \citep{Everitt2012}.
Within the Monte Carlo context, these types of distributions are referred to as \emph{doubly-intractable}.
\cite{Moller2006} introduced an ingenious method to address such problems for the Ising model.
For a more recent work, and for a concise but detailed review of developments since \cite{Moller2006}, see \cite{Moores2020} and references therein. 

Additionally, $J$ is a parameter of a prior distribution, thus under the Bayesian analysis framework it would not be too unreasonable for it to be pre-specified.
As such, in what follows, we will treat $J$ as known.
For example, in the numerical studies below we use a value which gives good performance in all our pilot studies, with a preference for choosing a smaller value of $J$ to avoid imposing too much spatial structure via this prior distribution.

\section{Existing Monte Carlo Methods}
\label{sec:existing-monte-carlo}

Having established the problem of interest and reviewed the relevant distribution for encoding spatial information, we now turn our attention to a brief review of existing computational methods for investigating the characteristics of such complex distributions.
We begin by looking at the relatively simple Gibbs sampling method for the Potts model.
However, ultimately we would like to sample from, and subsequently characterise, a posterior distribution that consists of the Potts model as a prior, together with an intractable likelihood (detailed in Section \ref{section:Method}).
To this end, a common and standard approach to sampling
from densities expressed in terms of intractable integrals, known as the pseudo-marginal method, is also discussed.
In this paper, we use the SMC normalising constant estimator, to approximate said intractable integrals.
Thus, the section ends with a brief discussion on the SMC estimator and its properties.
However, generic pseudo-marginal methods would be computationally infeasible in the context of high-dimensional discrete latent spaces of the sort considered here; this then leads to the
proposal of a novel computational method in the subsequent section.

\subsection{ Gibbs Sampler for the Potts Model}\label{sec:gibbs-sampler-potts} 
There exist efficient samplers for the Potts model, most notably Glauber Dynamics \cite{Glauber1963} corresponding to single-site Gibbs updates, and the elegant Swendsen-Wang algorithm \citep{Swendsen1987}.
Typically, the Swendsen-Wang algorithm is used in simple Potts models, where there is an absence of a ``strong external filed'' i.e. the (graph) likelihood in this context.  In this regime, the Swendsen-Wang algorithm has appealing convergence properties for all values of the coupling strength, $J$, see \cite{Hurn2003}.
However, this excellent performance deteriorates markedly in the presence of an external field \citep{Higdon1998} and hence is less appealing in the presence of an informative likelihood. In settings in which the likelihood is relatively uninformative, there might be benefits in combining this type of cluster update with the local moves explored here.

We focus on the Gibbs sampling approach \citep{Glauber1963,Geman1984}, variants of which will extend naturally to the context of interest in this paper.
The full conditionals required for a Gibbs sampler are readily computed for the Potts model:
$$ p(x_v \vert \bm{x}_{-v}) = \frac{\exp\left( J \sum_{u \in \partial(v)} \delta_{x_u,x_v}
  \right) }{ \sum_{x' \in \mathcal{X}} \exp\left( J \sum_{u \in \partial(v)}
    \delta_{x_u, x'} \right)} .$$

As is typical in Gibbs samplers, the use of full conditionals results in a simple node--wise update schedule.

However, computing full conditionals for posterior densities such as that of interest here is not possible because the likelihood cannot be evaluated even up to a normalising constant.
So we turn to another Monte Carlo method to tackle this problem. 

\subsection{Pseudo-Marginal Monte Carlo Markov Chain}
\label{sec:PM monte carlo markov chain}
Consider, first, MCMC in the simpler case of a tractable target density; one which can be evaluated point-wise.

Denote the target probability density $\mu(x)$,
where $x \in \mathcal{X}$, and $\mathcal{X} \subseteq \mathbb{R}^d$ is the state
space. MCMC algorithms construct a $\mu-$invariant Markov chain denoted
$(X^{(i)})_{i=1}^{n}$ and like all time-homogeneous Markov chains are completely characterised by their initial distribution and their transition kernel.

The Metropolis-Hastings (MH) \citep{Metropolis1953, Hastings1970} algorithm is a
particularly widely-used MCMC algorithm that may be used to target any
analytically tractable density $\mu$. The algorithm can be described very
briefly as follows:

Given that the MH Markov chain is at some state $x\in\mathcal{X}$, a new
state $x^*\in\mathcal{X}$ is proposed from a proposal density, denoted here
$q(x,\cdot)$. The proposed state $x^*$ is accepted with
probability $ \min \{1,R(x,x^*)\}$, where the acceptance ratio is defined
$$ R(x,x^*): = \frac{\mu(x^*)q(x^*,x)}{\mu(x)q(x,x^*)}; $$
or remains in the current state otherwise. The MH algorithm specifies the kernel to
be of the form:
$$ K_{\text{MH}}(x,\dd x^*):= \min \{1,R(x,x^*)q(x,x^*)\}+\alpha_x\delta_x(\dd x^*),$$
where,
$$ \alpha_x := 1 - \int \min\{ 1,R(x,x^*) \} q(x,x^*) \dd x^*,$$
is the probability of remaining in the same state and $\delta_{x}$ denotes a probability measure concentrated at $x$.

The requirement of $\mu$ being known point--wise arises when computing $R$. Thus, a very natural approach for intractable target
densities is to approximate $\mu$ instead.
\citet{Beaumont2003} and \cite{Andrieu2009} show that as long as the
approximation of $\mu$ is unbiased, the Markov chain of this approximate version will have the same invariant distribution as when the ``exact'' target distribution had been used. This approximation results in a class of techniques called ``pseudo-marginal'' algorithms.

Many intractable target densities, including those of interest here, can be expressed in terms of marginals of tractable densities: integrals that cannot be evaluated analytically. An obvious example is when we wish to use marginal likelihoods as will be the case in our context; see~\eqref{eq:PosteriorDensity} below.

More formally, let the target density be of the form
$$ \mu(x) = \int_\Theta \mu(x,\theta) \dd\theta,$$
where the integrals cannot be solved analytically; and 
$\theta \in \Theta$ may be considered a latent (or nuisance) variable that is not necessarily of
current interest.

Let $\widehat{\mu}(x)$ be an estimator of $\mu(x)$ and suppose that $\widehat{\mu}(x)$
is unbiased, for all $x \in \mathcal{X}$.
That is, if we let $g(\cdot \vert x)$ denote the density of $\widehat{\mu}(x)$, then we have that $\mathbb{E}[\widehat{\mu}(x)] = \mu(x)$ for every $x\in\mathcal{X}$,
where the expectation is taken with respect to $g(\cdot \vert x)$.
Note that $g$ need not be known.

The estimator $\widehat{\mu}(x)$ is random, and its variance will play an important part in producing accurate results --- we postpone further discussion until Section~\ref{section:NodewiseAlgoSection}.
The normalising constant estimator of the sequential Monte Carlo sampler, as detailed in Section \ref{sec:SMC Samplers}, is a common and popular marginal likelihood estimator. 

Grouped independence Metropolis Hastings (GIMH), a pseudo-marginal algorithm
first presented by \cite{Beaumont2003} and subsequently interpreted by \citet{Andrieu2009}, can now be presented. Using the simpler reformulation (without auxiliary variables) given by \citet{Andrieu2015}, the pseudo-code description of GIMH is presented in Algorithm \ref{algo:GIMHAlgorithm}.

\begin{algorithm}
\caption{ The GIMH Pseudo-Marginal Algorithm}\label{algo:GIMHAlgorithm}
\begin{enumerate}
\item Given $x^{(i-1)}$ and unbiased estimate
  $\widehat{\mu}_{(i-1)}(x^{(i-1)})$.
\item Sample:
\begin{enumerate}
\item $x^{*} \sim q(x^{(i-1)}, \cdot)$
\item $\widehat{\mu}_*(x^{*}) \vert x^* \sim g(\cdot \vert x^{*} )$
\end{enumerate}
\item Compute \[ \widehat{r}(x^{(i-1)},x^{*}):=
    \frac{\widehat{\mu}_{*}(x^*)q(x^*,x^{(i-1)})}{\widehat{\mu}_{(i-1)}(x^{(i-1)})q(x^{(i-1)},x^*)} \]
\item With probability $\min\{1,\hat{r}\}$ let:
\[ x^{(i)} = x^{*} \text{ and } \widehat{\mu}_{(i)}(x^{(i)}) = \widehat{\mu}_{*}(x^{*}); \hspace{1cm} \]
otherwise
\[ x^{(i)} = x^{(i-1)} \text{ and } \widehat{\mu}_{(i)}(x^{(i)}) = \widehat{\mu}_{(i-1)}(x^{(i-1)}).\hspace{1cm} \]
\end{enumerate}
\end{algorithm}
Here, we have used subscripts $(i)$ and $*$ for $\widehat{\mu}$ to emphasise precisely where the marginal likelihood has been (re-)estimated.
Henceforth these subscripts will be suppressed for notational clarity.

Note in particular, Algorithm \ref{algo:GIMHAlgorithm} is analogous to the marginal MH algorithm,
with the difference being that the acceptance ratio, \(R\), is now approximated by
$$ \widehat{R}(x,x^*): =
  \frac{\widehat{\mu}(x^*)q(x^*,x)}{\widehat{\mu}(x)q(x,x^*)}.$$
  As \citet{Andrieu2009} point out, the pseudo-marginal Markov chain should be intuitively thought of as a Markov chain on the extended space of the random variable $\left(X, \widehat{\mu}(X)\right)$, rather than just $X$ itself.
They further show that a Markov kernel with this acceptance ratio produces a Markov chain with (and converging to) invariant distribution with marginal  $\mu$.
A similar justification applies to the algorithm proposed below which is a natural extension of GIMH specialised to the setting of interest.
We detail this in Section \ref{sec:exact-samples-from}, Proposition \ref{prop_1}.  

\subsection{ Estimating the marginal likelihood using an SMC Sampler}
\label{sec:SMC Samplers}
Unbiased likelihood estimators are commonly constructed using particle filters, or SMC samplers, as in \citet{Andrieu2010}.

SMC samplers \citep{ DelMoral2006} are a class of algorithms that generate collections of weighted samples approximating
each of a sequence of target distributions $\{ \mu_t\}_{t\geq 1}$  over  essentially any random variable defined on some spaces
$(\mathcal{X}_t)_{t\geq 1}$.
Standard SMC methods achieve this through a combination of
sequential IS and variance reduction resampling methods on spaces with increasing
dimensions (see, for example, \citet{Doucet2011}). However, SMC
samplers allow for sequence of distributions defined on some common space
$\mathcal{X}$.

SMC samplers accomplish this by creating a sequence of auxiliary
distributions
on spaces of increasing dimensions.
SMC samplers can also be readily applied to simpler spaces through the use of annealing schemes.
See \cite{Zhou2016} and Section \ref{section:EmStudy} for details on specification of the annealing scheme within the present Bayesian analysis context. 
The sequence of auxiliary distributions are specified
using Markov kernels called ``backward kernels'', which critically  impact the variance of the estimators. See \cite{DelMoral2006} for more details on these so-called backward kernels.

Conventional sequential importance re-sampling algorithms can then be applied to
the auxiliary distribution sequence. More precisely, assume that at iteration $t-1$ the set of
$N$ weighted particles approximate $\widetilde{\mu}_{t-1}$, denote this
$\{W_{t-1}^{(i)},X_{0:t-1}^{(i)}\}_{i=1}^{N}$.
Here, $W_{t}^{(i)}$ denotes the \emph{normalised} weight of particle $X_{0:t}^{(i)}$ for all $i = 1, \dots, N$.
For iteration $t$, the path of
each particle $X_{0:t-1}^{(i)}$ is extended with a Markov kernel say,
$K_t(x_{t-1},x_t)$, yielding the set of particles $\{X_{0:n}^{(i)}\}_{i=1}^{N}$
to which IS is applied. Subsequently, the weights are updated by a factor
called incremental weights.

Importantly, particularly when interested in approximating marginals such as the
model evidence (marginal likelihood); if $\mu_t(x_t) = {\gamma_t(x_t)}/{Z_t}$ is only known up to
a normalising constant, the unnormalised incremental weights
$$w_t(x_{t-1},x_{t}) = \frac{\gamma_{t}(x_{t-1})}{\gamma_{t-1}(x_{t-1})}.$$
can be used.
Note: here, and as practised in the studies below, we assume that $K_t$ is constructed to be $\mu_t$-invariant,
$\mu_t \ll \mu_{t-1}$ and the associated time-reversal kernel is used for the backward kernel.

In fact, this collection of algorithms provide unbiased (see \citet[Proposition 7.4.1]{DelMoral2004}) estimates of ${Z_t}/{Z_{t-1}}$ via
\begin{align}
  \widehat{\frac{Z_t}{Z_1}} = \prod_{p=2}^{t}\widehat{\frac{Z_p}{Z_{p-1}}} = \prod_{p=2}^{t}\sum_{i=1}^{N}W^{(i)}_{p-1}w(X^{(i)}_{p-1:p}).\label{eq:smc_norm_const} 
  \end{align}

  These normalising constant estimators satisfy a Central Limit Theorem \citep[Proposition 9.4.1]{DelMoral2004} and have asymptotic variance that are inversely proportional to the number of particles $N$ (see also \cite{Chopin2020} for a recent review).
  As such, $N$ is an important tuning parameter when this sampler is used within a pseudo-marginal algorithm \citep{Doucet2015}.

\section{Methodology}\label{section:Method}

We now turn our attention to the development and presentation of the models to encode spatial dependence together with computational methods that can be used to perform inference under, and characterise, these models.
 
There is a considerable literature on spatial modelling in Bayesian data analysis (see, for example, \citet{Gelfand2010}).
However, some specialised settings, such as that encountered when dealing with PET images, are yet to be explored to the same extent.
In these cases, there may be significantly large amounts of data to be analysed and relatively sophisticated models may be required for meaningful investigation.
Such approaches may need to incorporate application-specific structure and information. For example, \cite{Bezener2018} presents a method of variable selection which models spatial dependence using hierarchical spatial priors and parcellation for the analysis of MRI images.

Studies that attempt to model spatial relationships in PET images are scarce, and tend not to use a Bayesian approach.
One such example is \citet{Zhou2002} who proposes a method to allow for spatial dependence.
That is, spatial constraint using nonlinear ridge regression is used to improve upon parametric images produced by conventional weighted NLS methods.
They showed that doing so reduced the percentage mean square error by $60-80\%$ in simulation studies.
These promising results motivate the incorporation of spatial dependence when analysing PET images, particularly in Bayesian frameworks where such studies have been largely absent.

On the other hand, there exists many computationally efficient methods for complex settings, such as PET data analysis. These include, in the PET context, among many others: \citet{Gunn2002,Peng2008,Jiang2009, Zhou2013, Zhou2016}. Computational tractability in these methods is typically attained through the assumption of spatial independence.

A natural strategy is, therefore, to adapt existing non-spatial methods to a model that does incorporate spatial dependence. This strategy is particularly natural when modelling spatial dependence via a Potts model as it can be efficiently targeted by MCMC methods with single-site update schemes. In this section we construct a class of models that allows for such computational methods.

Having presented the model that allows for effective incorporation of spatial dependence, we then present a natural extension of the generic pseudo-marginal approach, discussed in Section \ref{sec:existing-monte-carlo} above, to this specialised setting in Section~\ref{section:NodewiseAlgoSection} along with the theoretical justification of the method, which follows from that of the standard pseudo-marginal approach and some additional considerations. Further extensions of the method are described in Section~\ref{sec:Mult-augM-pseudo}, and a straightforward approximation of the proposed algorithm is introduced in Section~\ref{sec:Approximations}.

\subsection{ Spatial Bayesian Model Selection using a Potts Prior}

Given spatial data (or spatio-temporal data, such as PET) $\bm{Y} = (Y_v : Y_v \in \mathbb{R}^k, v \in V )$, associate a graph $G_{\bm{Y}}=(V_{\bm{Y}},E_{\bm{Y}})$ that encodes the geometry of spatial dependence. 
That is, data points $Y_u$ is conditionally independent of $Y_v$ given $Y_{-\{u,v\}}$ unless $\langle u,v \rangle \in E_{\bm{Y}}$.
For example, for image data a simple (finite) lattice would suffice. We will call $Y_v$ the node data point at the node $v \in V$.

Building good statistical models for $\bm{Y}$ can be difficult; model selection would involve selecting among many complex and computationally intensive models from the model space.
Analysis in this setting can be simplified by using existing simpler models such as those that assume spatial independence and model data at the node-level.

Such an approach has two significant strengths.
Firstly, this approach can readily exploit many existing techniques which  make spatial-independence assumptions, allowing them to incorporate spatial dependence. Secondly, using this approach allows for considerable computational simplifications as will be seen later.

\subsubsection{A Generic Model for Spatial Model Selection}
\label{ssection:S-T Inference using Potts model}

Formally, consider parametric models for the node data point $Y_v$, for all $v \in V_{\bm{Y}}$, rather than the whole data set $\bm{Y}$.
Denote the parameter of these models by $\theta_v$, with parametric space $\Theta_v$ for each node $v \in V$. 

It is immediate that a Potts model can be used to encode spatial relationships between models at different locations in space; define 
$\bm{X} =( X_v \in \mathcal{X} : v \in V_{\bm{Y}})$ a Markov random field with respect to the graph $G_{\bm{Y}}$.
Here, the coupling constant $J$ is treated as known and constant.

The parametric distribution for $Y_v$, now has parameters $(\theta_v, X_v)$.
Here, $X_v$ is a component of the Potts random variable $\bm{X}$, and so has spatial dependence dictated by $E_{\bm{Y}}$.
The model order at node $v$ is given by $X_v$. Given a family of parameterised prior densities for each possible model denoted generically  $p$, over the parameter space; we may compactly summarise this hierarchical model as:
\begin{align}
  \bm{X} \vert J, G_{\bm{Y}}&\sim \text{Potts}(J, G_{\bm{Y}}), \label{eq:XModelPotts}\\
  \theta_{v} \vert X_{v}, \xi_v &\sim p_{X_{v}}(\cdot \vert \xi_v) \text{ for all }v \in V_{\bm{Y}},\label{eq:XModelPrior}\\
  Y_{v} \vert \theta_{v},X_{v} &\sim f_{X_{v}}(\cdot \vert \theta_v) \text{ for all }v \in V_{\bm{Y}}. \label{eq:XModelLikelihood}
\end{align}

Here, recall that $f_M$ denotes the density associated with statistical model $S_M$ with model order $M$.
A similar convention is used for the (parameter) prior to denote that a different prior distribution must be used for each model order in the Bayesian context.
Specifically, here we have that the model order is $M=X_v$ for each node $v \in V$.
This generic model is represented as a plate diagram in Figure \ref{fig:platesDiagram2}.

\begin{figure}
  \centering
  \tikzset{every picture/.style={line width=0.75pt}} 

\begin{tikzpicture}[x=0.75pt,y=0.75pt,yscale=-1,xscale=1]

\draw   (260,60) -- (430,60) -- (430,270) -- (260,270) -- cycle ;
\draw    (310,40) -- (310,78) ;
\draw [shift={(310,80)}, rotate = 270] [color={rgb, 255:red, 0; green, 0; blue, 0 }  ][line width=0.75]    (10.93,-3.29) .. controls (6.95,-1.4) and (3.31,-0.3) .. (0,0) .. controls (3.31,0.3) and (6.95,1.4) .. (10.93,3.29)   ;
\draw   (290,100) .. controls (290,88.95) and (298.95,80) .. (310,80) .. controls (321.05,80) and (330,88.95) .. (330,100) .. controls (330,111.05) and (321.05,120) .. (310,120) .. controls (298.95,120) and (290,111.05) .. (290,100) -- cycle ;
\draw   (290,180) .. controls (290,168.95) and (298.95,160) .. (310,160) .. controls (321.05,160) and (330,168.95) .. (330,180) .. controls (330,191.05) and (321.05,200) .. (310,200) .. controls (298.95,200) and (290,191.05) .. (290,180) -- cycle ;
\draw   (290,240) .. controls (290,228.95) and (298.95,220) .. (310,220) .. controls (321.05,220) and (330,228.95) .. (330,240) .. controls (330,251.05) and (321.05,260) .. (310,260) .. controls (298.95,260) and (290,251.05) .. (290,240) -- cycle ;
\draw   (350,179) .. controls (350,167.95) and (358.95,159) .. (370,159) .. controls (381.05,159) and (390,167.95) .. (390,179) .. controls (390,190.05) and (381.05,199) .. (370,199) .. controls (358.95,199) and (350,190.05) .. (350,179) -- cycle ;
\draw    (310,120) -- (310,158) ;
\draw [shift={(310,160)}, rotate = 270] [color={rgb, 255:red, 0; green, 0; blue, 0 }  ][line width=0.75]    (10.93,-3.29) .. controls (6.95,-1.4) and (3.31,-0.3) .. (0,0) .. controls (3.31,0.3) and (6.95,1.4) .. (10.93,3.29)   ;
\draw    (310,200) -- (310,218) ;
\draw [shift={(310,220)}, rotate = 270] [color={rgb, 255:red, 0; green, 0; blue, 0 }  ][line width=0.75]    (10.93,-3.29) .. controls (6.95,-1.4) and (3.31,-0.3) .. (0,0) .. controls (3.31,0.3) and (6.95,1.4) .. (10.93,3.29)   ;
\draw    (290,100) .. controls (290.61,94.89) and (256.96,170.11) .. (289.5,238.96) ;
\draw [shift={(290,240)}, rotate = 244.22] [color={rgb, 255:red, 0; green, 0; blue, 0 }  ][line width=0.75]    (10.93,-3.29) .. controls (6.95,-1.4) and (3.31,-0.3) .. (0,0) .. controls (3.31,0.3) and (6.95,1.4) .. (10.93,3.29)   ;
\draw  [dash pattern={on 4.5pt off 4.5pt}] (349,87.4) .. controls (349,80) and (355,74) .. (362.4,74) -- (412.6,74) .. controls (420,74) and (426,80) .. (426,87.4) -- (426,127.6) .. controls (426,135) and (420,141) .. (412.6,141) -- (362.4,141) .. controls (355,141) and (349,135) .. (349,127.6) -- cycle ;
\draw   (354,99) .. controls (354,87.95) and (362.95,79) .. (374,79) .. controls (385.05,79) and (394,87.95) .. (394,99) .. controls (394,110.05) and (385.05,119) .. (374,119) .. controls (362.95,119) and (354,110.05) .. (354,99) -- cycle ;
\draw    (353.58,100.08) -- (332,99.09) ;
\draw [shift={(330,99)}, rotate = 362.63] [color={rgb, 255:red, 0; green, 0; blue, 0 }  ][line width=0.75]    (10.93,-3.29) .. controls (6.95,-1.4) and (3.31,-0.3) .. (0,0) .. controls (3.31,0.3) and (6.95,1.4) .. (10.93,3.29)   ;
\draw    (350,180) -- (332,180) ;
\draw [shift={(330,180)}, rotate = 360] [color={rgb, 255:red, 0; green, 0; blue, 0 }  ][line width=0.75]    (10.93,-3.29) .. controls (6.95,-1.4) and (3.31,-0.3) .. (0,0) .. controls (3.31,0.3) and (6.95,1.4) .. (10.93,3.29)   ;

\draw (303,21) node [anchor=north west][inner sep=0.75pt]  [font=\large] [align=left] {$\displaystyle J$};
\draw (301,89) node [anchor=north west][inner sep=0.75pt]  [font=\large] [align=left] {$\displaystyle X_{v}$};
\draw (301,172) node [anchor=north west][inner sep=0.75pt]  [font=\large] [align=left] {$\displaystyle \theta _{v}$};
\draw (301,229) node [anchor=north west][inner sep=0.75pt]  [font=\large] [align=left] {$\displaystyle Y_{v}$};
\draw (371,244) node [anchor=north west][inner sep=0.75pt]  [font=\large] [align=left] {$\displaystyle v\ \in V$};
\draw (361,171) node [anchor=north west][inner sep=0.75pt]  [font=\large] [align=left] {$\displaystyle \xi _{v}$};
\draw (367,120) node [anchor=north west][inner sep=0.75pt]  [font=\small] [align=left] {$\displaystyle {\displaystyle u\ \in \partial ( v)}$};
\draw (365,88) node [anchor=north west][inner sep=0.75pt]  [font=\large] [align=left] {$\displaystyle X_{u}$};

\end{tikzpicture}
  \caption{Plate diagram of the proposed generic model.}\label{fig:platesDiagram2}
\end{figure}
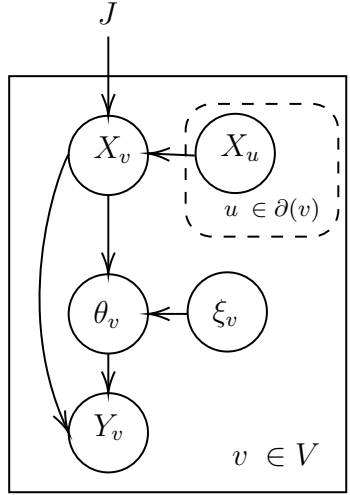

Here, $\xi_v$ is a parameter of the distribution $p$. In principle a hyperprior could be attached to this parameter with no particular difficulty; we work with this simple setting in the interest of parsimony.
To this end, we fix $\xi_v$ to be a known hyperparameter common to all nodes, $v$, and so we write $\xi$ instead of $\xi_v$ in what follows.

Additionally, given that we now treat the model order as a random variable, henceforth we use the convention that $f(\cdot \vert \theta_v,M) = f_M(\cdot \vert \theta_v)$.
Note that, at the nodal (pixel) level, Bayesian model selection would often involve the computation of the marginal likelihood, $f(y_v \vert M) = \int_{\Theta_v} f(y_v \vert \theta_v,M)p(\theta_v\vert M \xi)\mathrm{d} \theta_v$.
We are interested in the marginal likelihood of the whole image $\bm{Y}$, which may involve multiple intractable high-dimensional integrals.
We impose an assumption of conditional independence, discussed next, which allows us to make significant simplifications.
For a simple example, see the toy model described in Section \ref{section:NormalDistrToyExample};
This toy model will be used for simulation experiments to evaluate the proposed method.

In realistic settings, it may be that the parameter $\bm{\theta} = (\theta_v)_{v \in V}$ also exhibits spatial dependence.
The above hierarchical model, makes the underlying assumption that it does not and this setting is the focus of the present paper; where we have found the incorporation of spatial structure at the level of model sufficient to improve upon the state of the art for problems of interest. Further generalisation would be possible and provides an interesting direction for further work.
We encode this in Assumption~\ref{ass:cip} which will be a standing assumption throughout this paper and is discussed in Section~\ref{sec:dcip}.

\begin{assumption}[Conditional Independence]\label{ass:cip}
  $$\theta_v \vert X_v,\xi \perp \bm{X}_{-v}, \bm\theta_{-v} \textrm{ for all }v \in V.$$
\end{assumption}

The primary aim of the proposed methodology is to infer $\bm{X}$.
Although $\bm{\theta}$ \emph{is} inferred as a by-product of the proposed method,
it can be thought of as a latent variable in this framework and may be a nuisance parameter in some settings which mitigates the impact of modelling the parameters in this way. 

Finally, note that, as a direct consequence of the above model and assumption, the marginal likelihood of $\bm{Y}$ can be written as the product of the marginal likelihoods of $Y_v$ over all $v \in V$ --- see Section~\ref{sec:infer-froM-prop-1} below.

\subsubsection{ Spatial Bayesian Model Selection}
\label{sssection:S-T model selection with Potts}

To exploit the presented model, denote by $\mathcal{S}_v = \{S_{v,M_v} : M_v \in \mathcal{M}\}$ the model space used at each nodal data point $Y_v$. I.e. every node in the graph is associated with a model $M_v$ from a set common to all nodes, $\mathcal{M}$.
Thus a statistical model is associated with each node in the graph and hence each data point,
$$S_{v,M_v} = \{f(\cdot \vert {\theta_v},M_v): \theta_v \in \Theta_{v,M_v} \}.$$

The notation $\Theta_{v,M_v}$ is used here to emphasis that since $M_v$ is the model order at node $v$, the parameter space will be dependent on it.
For example, for compartmental models, the model order dictates the dimension of the parameter space.
These collections of models can be used to generate a model space for the whole data set $\bm{Y}$. Specifically, a set of candidate models for $\bm{Y}$ can be formulated as
$$\mathcal{S}_{\bm{Y}} = \{(S_{v,M_v}) : M_v \in \mathcal{M}\}_{v \in V}.$$

Writing $\bm{M} = (M_v)_{v \in V} \in \mathcal{M}^{V}$, Bayesian model selection of spatial data $\bm{Y}$ can be thought of as inference of the model order parameter $\bm{M}$ in the space $\mathcal{M}^{V}$.
Each realisation of $\bm{M}$ is called a configuration, note that there are $ \vert \mathcal{M} \vert ^{ \vert V \vert }$ candidate models for $\bm{Y}$.

As mentioned before, the Potts model is a natural choice for the prior distribution over model order. For model selection, we adapt the hierarchical model above such that:

\begin{enumerate}
\item  $\bm{M} = (M_v)_{v \in V} $ is a Markov random field with a Potts distribution, with spatial dependence represented by $E_{\bm{Y}}$ as before;
\item  Each $\theta_v$, for all $v \in V$, is dependent on $M_v$ (e.g. model order determines parameter dimension in some cases, accordingly we denote the prior $p(\cdot \vert \xi,M_v)$) and known constant hyper-parameter $\xi$, importantly  it is spatially independent (given the model order $M_v$);
\item  $Y_v$, the observed data, has likelihood $f(\cdot \vert \theta_v,M_v)$, where the model $M_v$ dictates which model is selected, for all $v \in V$. 
\end{enumerate}

Formally, the above can be summarised:
\begin{align}
 \bm{M} && \sim& \text{Potts}(J, G_{\bm{Y}}) && \label{eq:PottsPrior} \\
\theta_v& \vert M_v,\xi & \sim& p(\cdot;M_v, \xi) \text{ for all }v \in V, &&\label{eq:ParametersPrior}\\
Y_{v}& \vert \theta_v, M_v &\sim& f(\cdot; \theta_v,M_{v}) \text{ for all } v \in V. &&\label{eq:YLikelihood}
\end{align}

\subsubsection{Discussion of Assumption~\ref{ass:cip}}\label{sec:dcip}
Assumption~\ref{ass:cip} may not hold in general: if there is spatial dependence between the generative model which describes each vertex then there may also be dependence between the parameters of those models.
In some settings it can be viewed as an approximation which may be tolerable in order to allow inference under a model which is, at least, better than the assumption of full spatial independence.
In other settings, this assumption may not be unrealistic. It is also noteworthy that, if the model order dictates the dimensions of the parameter space at that node, and the space contains dependent parameters, it may be difficult to incorporate spatial dependence at a the parameter level in a meaningful manner.

For example, if considering compartmental models: if two adjacent pixels were to contain different number of compartments, it may be difficult to say anything meaningful about the spatial dependence of the micro-parameters such as the transfer rates.
Such considerations make it difficult to impose some more general Markov random field over the parameters.

Furthermore, incorporating spatial dependence at a model order level will typically take some precedence over spatial dependence at the parameter level.
For instance, in the PET setting, generally we would like to know if two adjacent pixels have the same compartments (i.e. whether or not the localised area contains receptors) before we think about the transfer rates.

To summarise the problem of interest in general terms: the image $\bm{Y}$ is modelled node-wise using the density $f(\cdot;\theta_v, M_v)$, $v \in V_{\bm{Y}}$, with priors over the parameters $\bm{\theta}$ and $\bm{M}$.
In particular, we propose the use of the Potts model as a prior over discrete parameters $\bm{M}$ to allow for tractable incorporation of spatial dependence in images.
As mentioned above, $\bm{\theta}$ will be treated as a latent variable --- this is discussed in detail next. 

\subsubsection{Inference from the Proposed Model}
\label{sec:infer-froM-prop-1}

Having specified this class of models, we now turn to the task of inference.
Let $p(\bm{M})$ denote the Potts prior probability mass function over $\bm{M} \in \mathcal{M}^{V}$.
Then, the model posterior density, for $\bm{Y}$, is denoted
$$\pi(\bm{M} \vert \bm{y}) \propto f(\bm{y} \vert \bm{M})p(\bm{M}).$$

Here, $f(\bm{y} \vert \bm{M})$ will be called the \emph{graph} marginal likelihood of $\bm{y}$; it encodes a generative description of the data under the proposed model.

The proposed method allows us to incorporate spatial dependence at the level of the discrete parameter $\bm{M}$ only.
Subsequently, for computational tractability we impose Assumption~\ref{ass:cip}: given the model order $M_v$ the random variable $Y_v$ at each node $v\in V$ is independent.
In other words, as per \eqref{eq:XModelPrior} and \eqref{eq:XModelLikelihood}, we have that
$$ f(\bm{y} \vert \bm{M}) = \prod_{v \in V} f(y_v \vert M_v).$$
Here $f(y_v \vert M_v)$ denotes the likelihood at each node $v$, and will be henceforth called the \emph{node-wise} likelihood.
This assumption is important as it makes the graph likelihood tractable, and provides computational simplifications discussed later.
Typically, as described above, a parametric model is used, thus $f(y_v \vert M_v)$ will in fact be a marginal likelihood or evidence (i.e. an integral, which is likely to be analytically intractable).

We finally have the probability mass function of primary computational interest
\begin{align*}
\pi(\bm{M} \vert \bm{y}) \propto& \left\{ \prod_{v \in V} f(y_{v} \vert M_{v})\right\}p(\bm{M})\\
 \propto&\left\{ \prod_{v \in V} \int_{\Theta_v} f(y_{v} \vert \theta_v,M_v)p(\theta_v \vert M_v, \xi)\dd\theta_v \right\}\\&\times \exp\left(J \sum_{v\sim u} \delta_{M_v,M_u}\right)\numberthis \label{eq:PosteriorDensity}.
\end{align*}

Recall the mild assumption that hyper-parameter $\xi$ is known and henceforth suppressed in notation.

Even in this general formulation, this mass function is not particularly attractive.
There are some positives:
It is straightforward to evaluate up to the intractable normalising constant; and Assumption~\ref{ass:cip} does seemingly lead to a simpler computation even at this stage, as only the node-wise marginal likelihoods $f(y_v \vert M_v)$ need to be computed rather than the high-dimensional integration of the full marginal likelihood $f(\bm{y} \vert \bm{M})$.

However, even in this simpler setting, the graph marginal likelihood poses two computational difficulties: i) it is a product of $ \vert V \vert $ integrals and ii) these integrals in general are analytically intractable in most cases of interest.

\subsection{The Node-Wise Pseudo-Marginal Algorithm}\label{section:NodewiseAlgoSection}

We now turn to extending the techniques discussed in Section \ref{sec:existing-monte-carlo}, to build the novel computational method for use in the class of models constructed above.

The aim is to characterise the posterior density, $\pi(\bm{M} \vert \bm{y})$, of the model given in \eqref{eq:PottsPrior}-\eqref{eq:YLikelihood}.
Suppose, for now, that a generic MH algorithm could be used in this setting.
That is, to generate a MH Markov chain that targets the posterior density.
This then gives the acceptance ratio: $$R(\bm{M},\bm{M}^*) = \frac{\pi(\bm{M}^* \vert \bm{y})q(\bm{M}^*,\bm{M})}{\pi(\bm{M} \vert \bm{y})q(\bm{M},\bm{M}^*)}. $$

Even with a tractable marginal likelihood, computing this acceptance ratio would be costly. Under Assumption~\ref{ass:cip} (given $\bm{M}$ the parameters $\bm{\theta} = (\theta_v)_{v \in V}$, of the model over the data $\bm{y}$, are independent) the computational complexity of the problem is reduced considerably; now, only the node-wise marginal likelihoods need to be computed.
However, every time a new configuration is proposed, up to $ \vert V \vert $ new integrals would then need to be computed and the acceptance probability is likely to be very small. Given that the number of possible Potts configurations is $ \vert \mathcal{M} \vert ^{ \vert V \vert }$, growing geometrically with the number of nodes (i.e. pixels, for image data), the computational costs are not appealing.

Under Assumption~\ref{ass:cip} it is natural to employ a Metropolis--within--Gibbs approach to mitigate these difficulties and to obtain significant computational simplifications.
Suppose that the MCMC sampler is currently at some configuration $\bm{M}=(M_v)_{v \in V}$.
Under a single variable updating schedule: \emph{only} the model order $M_v$, at some node $v\in V$, will be proposed to be changed in a given step.
We will refer to this as node-wise updating schedule, since we update the model order at a single node $v$ only.
Denote the proposed state to be $M^*_v$; and $\bm{M}^*_{[v]}$ to be the vector equivalent to $\bm{M}$ but with $M^*_v$ as the $v$-th coordinate.
For simplicity, assume that the proposal distribution is symmetric, and the acceptance ratio is
\begin{align*}
  \frac{\pi(\bm{M}_{[v]}^* \vert \bm{y})}{\pi(\bm{M} \vert \bm{y})}  =& \frac{f(\bm{y} \vert \bm{M}_{[v]}^*)p(\bm{M}_{[v]}^*)}{ f(\bm{y} \vert \bm{M})p(\bm{M})}\\
                                                         =& \frac{f(y_v \vert M^*_v)p(\bm{M}^*_{[v]})}{f(y_v \vert M_v)p(\bm{M})}\\
                                                         =&\frac{\int f(y_v \vert \theta_{v},M^*_v)p(\theta_v \vert M_v^*)\dd\theta_vp(\bm{M}_{[v]}^*)}{\int f(y_v \vert \theta_v,M_v)p(\theta_v \vert M_v)\dd\theta_vp(\bm{M})}.
\end{align*}
Thus, the acceptance probability of each Metropolis-within-Gibbs step requires the computation of only a single integral. A complete sweep over the graph thus requires $ \vert V \vert $ such integrals but one could expect a reasonable proportion of these proposed moves to be accepted, in contrast to a global proposal which sought to update every node simultaneously.

Typically, even the node-wise marginal likelihood $f(y_v \vert M_v)$ will be difficult or impossible to evaluate analytically. A pseudo-marginal MH algorithm is a natural choice in such cases and is explored in the next section.

\subsubsection{Graph Model Selection using Node-Wise Likelihood Marginal Estimates}
\label{sec:whole-graph-model}
The node-wise update schedule, gives rise to nested iterations.  For clarity, term the outer iteration, indexed with $i$, of a full pass through the whole graph as the graphical iteration.
The inner iteration over $V$ will be termed the node-wise iteration.

Suppose that the MH Markov chain denoted, $(\bm{M}^{(i)})_{i =1}^{n}$, is at graphical iteration $i$ and that the state at node $v$ is being proposed to be changed.
Next, since the node-wise marginal likelihood is a normalising constant, for simplicity denote
$$Z_{v}(M) = f(y_v \vert M_v=M)\text{ for all } M \in \mathcal{M}.$$
That is, $Z_v(M)$ is the normalising constant at node $v$ when the model order $M_v=M$.

At graphical iteration $i$, at node $v$, for given model order proposal $M^{(i)}_v$, let
$(\widehat{Z}_v)^{(i)}$ denote the approximation of the marginal  $f(y_v \vert M_v=M^{(i)}_v).$
Associate with this generic estimator, the computation cost parameter $N$ (for the SMC sampler this was the number of particles).
As mentioned before this is a tuning parameter for the pseudo-marginal, thus also the algorithm presented below.

It is important to note the slight abuse of notation here --- $(\widehat{Z}_v)^{(i)}$ is dependent on the model order proposal $M_v^{(i)}$.
As previously stated, we require that $(\widehat{Z}_v)^{(i)}$ is unbiased.
Additionally, recall we denote $g(\cdot \vert M_v = M_v^{(i)})$ the possibly unknown density of $(\widehat{Z}_v)^{(i)}$.

Following the node-wise updating schedule, whenever a new state $M_v^* \sim q(M_{v}^{(i-1)}, \cdot)$ is proposed, a new node-wise marginal estimation, denoted $(\widehat{Z}_v)^{*}$, is recomputed. 
The node-wise acceptance ratio,
$\widehat{r}(\bm{M}^{(i-1)},\bm{M}_{[v]}^{(i-1)*})$,
is then $$\frac{p(\bm{M}^{*}_{[v]})(\widehat{Z}_v)^{*}q(M_v^*,M_v^{(i-1)})}{p(\bm{M}^{(i-1)})(\widehat{Z}_v)^{(i-1)}q(M_v^{(i-1)},M_v^*)}.$$
We emphasise that the computation of $\widehat{r}$, precisely speaking, involves not just the standard MH proposal of a new configuration in $\mathcal{M}^{V}$, but also a proposal of the random variable $\widehat{Z}_v$: it is
a pseudo-marginal algorithm in the sense of \cite{Andrieu2009}, with the added subtlety that the unbiased estimates of the marginal likelihood associated with every other node in the graph are also retained as a part of the extended state.
The justification for using this acceptance ratio will be detailed in Section \ref{sec:exact-samples-from}

The above choices and assumptions lead to the algorithm which we term the ``Node-Wise Pseudo-Marginal'' (NWPM) algorithm; presented in pseudo-code in {Algorithm \ref{algthm: Node-wise PM}}.

\begin{algorithm}
\caption{ The Node-Wise Pseudo-Marginal Algorithm}\label{algthm: Node-wise
  PM}
Given: (i) target density $\mu(\bm{M}) := \pi(\bm{M} \vert \bm{y})$; (ii)
unbiased node-wise marginal likelihood estimators $\widehat{Z}_v(M) \approx
\int f(y_v \vert M,\theta_v)p(\theta_v \vert M,\xi)\dd \theta_v $ for all  $M \in \mathcal{M}$, $v \in V$;
and (iii) coupling constant $J$, tuning parameter $N$.
\begin{enumerate}
\item At $i=1$:
\begin{enumerate}
\item Initialise $\bm{M}^{(1)}$.
\item Sample $(\widehat{Z}_v)^{(1)} \sim g(\cdot \vert M_v^{(1)})$ for each $v \in V$.
\end{enumerate}
\item For $i=2, \dots$:
\begin{enumerate}
\item For $v \in V$:
\begin{enumerate}
\item Sample:
$M_v^* \sim q(M_v^{(i-1)},\cdot).$
\item Sample:
$    (\widehat{Z}_v)^{*} \vert m_{v}^{*} \sim g(\cdot  \vert M_v^*).$
\item Compute:
  $$\hat{r} = \frac{p(\bm{M}_{[v]}^{*})(\widehat{Z}_v)^*q(M_v^*,M_v^{(i-1)})}{ p(\bm{M}^{(i-1)})(\widehat{Z}_v)^{(i-1)}q(M_v^{(i-1)},M^*_v,)}$$
\item With probability $\min\{1,\widehat{r}\}$ take:
$$M_v^{(i)} = M_v^*\text{ and } (\widehat{Z}_v)^{(i)} = (\widehat{Z}_v)^{*}$$
otherwise,
$$M_v^{(i)} = M_v^{(i-1)} \text{ and } (\widehat{Z}_v)^{(i)} = (\widehat{Z}_v)^{(i-1)}$$
\end{enumerate}
\end{enumerate}
\end{enumerate}
\end{algorithm}
Within the context of this paper, the NWPM algorithm proposed here is a broadly applicable extension of the GIMH algorithm described above in Algorithm \ref{algo:GIMHAlgorithm}.

The output of this algorithm, the Monte Carlo sample $(\bm{M}^{(i)})_{i=1}^{n}$, can be used to perform model selection in a principled manner.
For example, this could be done using the node-wise marginal modal model order:
i.e., at each node $v \in V$ the model order which occurs the most in the (marginal) chain $(M_v^{(i)})^{n}_{i=1}$ is selected.
Additionally, samples of the latent variables $\bm{\theta}$ are typically available as byproducts of the marginal likelihood estimates and can be used to perform parameter inference as demonstrated in the numerical studies below.

An advantage of this pseudo-marginal based approach in this setting is its flexibility. It can be used with essentially any unbiased estimator of marginal likelihood and hence allows existing technology from any application domain to be employed.
The incorporation of spatial dependence via the modelling framework and algorithm developed here can come at the very little computational cost by approximating or making pragmatic adaptations to Algorithm \ref{algthm: Node-wise PM}.
We employ such techniques in the simulation studies presented in Section \ref{section:EmStudy} and discus in more detail in Section \ref{sec:Approximations} below.

Finally, note that for convenience we assume that the tuning parameter $N$ is fixed to the same value for all $v \in V$.
This parameter typically determines the variance of the estimator $\widehat{Z}_v(M)$ and additional flexibility could be obtained by allowing it to vary between vertices.
Unsurprisingly, a larger $N$ should give a better performance at the expense of greater computational cost.
Indeed \cite{Andrieu2009}, and more recently \cite{Andrieu2016}, showed that when the dispersion of the marginal likelihood estimates are larger the mixing rate of any pseudo-marginal MH algorithm decreases.
Of course, there is a higher computational cost when using larger $N$ to reduce the estimate variance.
The effects of this tuning parameter, and the trade off,  on the mixing of the Markov chain in a standard pseudo-marginal algorithm have been studied by \citet{Sherlock2015,Doucet2015,Sherlock2016}.
These works give theoretical justification and numerical studies for the recommendation that $N$ should be chosen such that the variance of the log-likelihood estimator is close to one.

\subsubsection{Marginal invariant distribution of NWPM}
\label{sec:exact-samples-from} 
Here we establish the formal validity of a class of algorithms which includes the NWPM, allowing for more general moves than those described in Algorithm~\ref{algthm: Node-wise PM}. We show that this class of algorithms leads to a Markov chain with an invariant distribution coinciding with $\pi$, as given in \eqref{eq:PosteriorDensity}.
For further discussion on specifying proposals within this setting see Section \ref{sec:Mult-augM-pseudo}.

Recalling Assumption~\ref{ass:cip}, we may write the target distribution of the NWPM algorithm as 
$$ \mu(\bm{M}) \propto p(\bm{M})\prod_{v \in V} \int_{\Theta_v} f(y_v \vert \theta_v,M_v)p(\theta_v \vert M_v, \xi)\dd\theta_v ,$$
for $\bm{M}=(M_v)_{v \in V} \in \mathcal{M}^{V}$.

Given $\bm{M}$, denote the random vector of the unbiased node-wise marginal likelihoods estimators as $\widehat{\bm{Z}} \vert \bm{M} \sim g(\cdot \vert \bm{M})$, i.e.,
$$\bm{\widehat{Z}} = (\widehat{Z}_v(M_v):v\in V).$$
Henceforth using the shorthand $\widehat{Z}_{v}^M:=\widehat{Z}_{v}(M)$ for all $M \in \mathcal{M}$; we note that, due to independence,
$$g(\widehat{\bm{Z}} \vert \bm{M}) \propto \prod_{v \in V }g(\widehat{Z}_{v}^{M_v} \vert M_v).$$
Finally, define the extended target density 
\begin{align}
  \widehat{\mu}(\bm{M},\widehat{\bm{Z}})
  &= \frac{p(\bm{M}) \prod_{v \in V} \widehat{Z}_{v}^{M_v}g(\widehat{Z}_{v}^{M_v} \vert M_v)}{\sum_{\bm{M}'\in\mathcal{M}^V} 
    p(\bm{M}') \prod_{v \in V} f(y_v \vert M_v')}\\
  &= \frac{p(\bm{M})\prod_{v \in V}\widehat{Z}_{v}^{M_v}g(\widehat{Z}_{v}^{M_v} \vert M_v)}{f(\bm{y})}.\label{eq:pM_target_density} 
\end{align}

Here, we let $f(\mathbf{y})$ denote the marginal probability of the observed data, integrating out unknown parameters and unknown model orders; Given observation $\bm{Y} = \bm{y}$, we treat it as a normalising constant.
The NWPM algorithm has similar theoretical properties as GIMH. As such, the remainder of the argument of formal justification then follows the same as the GIMH case, see \cite{Andrieu2009}.
\begin{proposition}\label{prop_1} 
Let $\bm{M}$ denote a random variable in the graph model order space $\mathcal{M}^{V}$.
Let $\bm{\widehat{Z}} \vert \bm{M}$, $g$, $\mu$ and $\widehat{\mu}$ be defined as above.

For any $U \subset V$, let $q_U$ denote a Markov kernel on $(\mathcal{M} \times \mathbb{R}_+)^V$ of the form
\begin{align*}
  q_U((\bm{M},\bm{\widehat{Z}}),(\bm{M}^*&,\bm{\widehat{Z}}^*)) = \delta_{(\bm{M}_{-U},\bm{\widehat{Z}}_{-U})}(\bm{M}_{-U}^*,\bm{\widehat{Z}}^*_{-U})\\&\times
  q^{\mathcal{M}}_U(\bm{M}_U,\bm{M}_U^*) \prod_{v\in U} g(\widehat{Z}^*_v \vert M_v^*) ,
\end{align*}
where $q^{\mathcal{M}}_U$ denotes a Markov kernel on $\mathcal{M}^U$ and we slightly abuse density notation using the Dirac delta functions to indicate that variables associated with nodes outside $U$ are unchanged (absolute continuity of the numerator and denominator of the Metropolis-Hastings ratio is ensured by the symmetry of this singular part of the kernel).

The standard Metropolis-Hasting acceptance probability with target distribution $\widehat\mu$ can be expressed as:
\begin{equation}
  1 \wedge \frac{p(\bm{M}^*)(\prod_{v \in U}\widehat{Z}_v^* )q_U^{\mathcal{M}}(\bm{M}_U^*,\bm{M}_U)}{p(\bm{M})  (\prod_{v \in U}\widehat{Z}_v) q_U^{\mathcal{M}}(\bm{M}_U,\bm{M}_U^*)}
\end{equation}
and the marginal distribution of $\bm{M}$ under $\widehat\mu$ is $\mu$.
\end{proposition}


\begin{proof}
  First, consider the acceptance ratio of a MH (Metropolis-Hastings) algorithm with target distribution $\widehat\mu$ and proposal kernel $q_U$:
  \begin{align*}
    \hat{r} =& \frac{\widehat{\mu}(\bm{M}^*,\widehat{\bm{Z}}^*) q_U((\bm{M}^*,\widehat{\bm{Z}}^*), (\bm{M},\widehat{\bm{Z}}))}{\widehat\mu(\bm{M},\widehat{\bm{Z}}) q_U((\bm{M},\widehat{\bm{Z}}), (\bm{M}^*,\widehat{\bm{Z}}^*))}
  \end{align*}
  upon inserting the definition of $q_U$ one observes that the numerator is absolutely continuous with respect to the denominator and the singular elements simply impose that $\bm{M}_{-U} = \bm{M}^*_{-U}$ and $\widehat{\bm{Z}}_{-U} = \widehat{\bm{Z}}_{-U}^*$, and we have:
  \begin{align*}
    \hat{r} =& \frac{\widehat{\mu}(\bm{M}^*,\widehat{\bm{Z}}^*) q_U^{\mathcal{M}}(\bm{M}^*_U, \bm{M}_U) \prod\limits_{v\in U} g(\widehat{Z}_v \vert M_v)}{\widehat{\mu}(\bm{M},\widehat{\bm{Z}}) q_U^{\mathcal{M}}(\bm{M}_U, \bm{M}^*_U) \prod\limits_{v\in U} g(\widehat{Z}_v^* \vert M_v^*)}\\
    =& \frac{p(\bm{M}^*)}{p(\bm{M})} \prod_{v\in V} \frac{\widehat{Z}_v^* g(\widehat{Z}_v^* \vert M_v^*)}{\widehat{Z}_v g(\widehat{Z}_v \vert M_v)}
    \\&\times
    \frac{q_U^{\mathcal{M}}(\bm{M}^*_U, \bm{M}_U)}{q_U^{\mathcal{M}}(\bm{M}_U, \bm{M}^*_U) } \prod\limits_{v\in U}\frac{ g(\widehat{Z}_v \vert M_v)}{g(\widehat{Z}_v^* \vert M_v^*)}\\
    =& \frac{p(\bm{M}^*)}{p(\bm{M})} \prod_{v \in U} \frac{\widehat{Z}_v^*}{\widehat{Z}_v}\\      &\times \frac{q_U^{\mathcal{M}}(\bm{M}^*_U, \bm{M}_U)}{q_U^{\mathcal{M}}(\bm{M}_U, \bm{M}^*_U) } \underset{=1}{\underbrace{ \prod_{v\not\in U} \frac{\widehat{Z}_v^* g(\widehat{Z}_v^* \vert M_v^*)}{\widehat{Z}_v g(\widehat{Z}_v \vert M_v)}}},
  \end{align*}

  where the final factor is equal to one almost surely under the proposal distribution, and the result follows.

  The marginal distribution of $\bm{M}$ follows by the essentially the same argument as in the standard pseudo-marginal context.
\begin{align*}
  &\int_{\mathbb{R}_+^{V}} \widehat{\mu}(\bm{M},\widehat{\bm{Z}}) \dd\widehat{\bm{Z}} \\
  & = \frac{p(\bm{M})}{f(\bm{y})}\prod_{v \in V} \int_{\mathbb{R}_+}\widehat{Z}_{v}^{M_v}g(\widehat{Z}_{v}^{M_v} \vert M_v)\dd\widehat{Z}_{v}^{M_v} \\
  &= \frac{p(\bm{M})}{f(\bm{y})}\prod_{v \in V}\mathbb{E}_g[\widehat{Z}_{v}(M_v)]\\
  &= \frac{p(\bm{M})\prod_{v \in V}f(y_v \vert M_v)}{f(\bm{y})}\\
  &= \mu(\bm{M}).
\end{align*} 
\end{proof}

Clearly, the node-wise proposals described previously fit this framework with $U=\{u\}$ being the single node being updated, although this framework would allow a somewhat broader class of proposals and blocked Metropolis-within-Gibbs type strategies to be explored. Standard arguments mean that a mixture or cycle of such kernels will also preserve $\widehat\mu$ as the invariant distribution.

One might anticipate that the theoretical properties of the NWPM algorithm might be further characterised by extending techniques from the standard pseudo-marginal setting such as \cite{Andrieu2015,Andrieu2016,Andrieu2021} to this one, indeed it would be surprising were it to behave substantially differently, but doing this rigorously would involve some delicate technical work and is beyond the scope of this paper.

Next, we discuss some immediate extensions and innovations of this proposed methodology; resulting in further algorithms, approximations and future approaches.

\subsection{Multiple Augmentation Pseudo-marginal Algorithms}
\label{sec:Mult-augM-pseudo} 

In this section we will explore the potential of further augmentation of the state space.
The specialised setting of the problem, allows approaches which to the authors' knowledge, have not been previously studied.
The strategy developed here allows for a number of further extensions.
Some simple approaches are discussed below --- methods based on these approaches will also be evaluated in the numerical studies.

The augmentation of the state space used to justify node-wise pseudo-marginal algorithms can be further extended by adding an estimate of the marginal likelihood associated with \emph{every} possible model at \emph{every} node to the state space. Although doing so may seem counter-intuitive and leads to a rather large state space, it allows a number of algorithmic innovations. In particular, by considering the following extended state space it will be possible to use a variety of standard MH moves in order to explore this space.

Specifically, the further extended space allows us to decouple the updating of the marginal likelihood estimates from those of the model configurations themselves; the standard pseudo-marginal update would require that these two updates occur simultaneously. We exploit this decoupling when we go on to consider the possibility of updating the marginal likelihood estimates with lower frequency than the model configurations. 

Let $\bar{\bm{Z}} := (\widehat{Z}_{v}(M): v \in V, M \in \mathcal{M})$ be a vector of marginal likelihood estimator for each model order $M \in \mathcal{M}$, at every node $v \in V$.
Consider, next, a further extended form of the target density \eqref{eq:pM_target_density} : 

$$\bar{\mu}(\bm{M},\bar{\bm{Z}}) = \frac{p(\bm{M})}{f(\bm{y})}\prod_{v \in V} \widehat{Z}_{v}^{M_v} \left(\prod_{m' \in \mathcal{M}} g(\widehat{Z}_{v}^{M'} \vert M')\right).$$
The extended joint density $\bar{\mu}$ is constructed such that its marginal over $\bm{M}$ coincides exactly with the \emph{correct posterior distribution}. Indeed, using essentially the same argument as in Proposition~\ref{prop_1}, we have:
\begin{align*}
  &\int_{\mathbb{R}^{V \times \mathcal{M}}_+} \bar{\mu}(\bm{M},\bar{\bm{Z}}) d\bar{\bm{Z}} \\
  &= \frac{p(\bm{M})}{f(\bm{y})}\prod_{v \in V} \int_{\mathbb{R}_{+}^{\mathcal{M}}}\widehat{Z}_{v}^{M_v}\left(\prod_{m' \in \mathcal{M}}g(\widehat{Z}_{v}^{M'} \vert M') \dd\widehat{Z}_v^{m'}\right) \\
    &= \frac{p(\bm{M})}{f(\bm{y})}\prod_{v \in V} \int_{\mathbb{R}_+}\widehat{Z}_{v}^{M_v} g(\widehat{Z}_{v}^{M_v} \vert M_v) \dd\widehat{Z}_{v}^{M_v} \\
  &=  \mu(\bm{M}).\\
\end{align*}

This further extended target distribution allows for some generalisations of the standard pseudo-marginal algorithm in the context of interest; it depends fundamentally on the fact that the variables associated with each node of the graph take values within a small finite set. It is possible to consider a variety of MH like moves applied to this extended target density. For simplicity we will consider only two types of proposal here: one which changes a single node's associated model and another which refreshes the likelihood estimates for a single node.


Denote the proposal of the random vector $\bar{\bm{Z}} = (\widehat{Z}^{M}_{v}: v \in V, M \in \mathcal{M})$ by $\bar{\bm{Z}}^*:=((\widehat{Z}^{M}_v)^*: v \in V, M \in \mathcal{M})$.
Intuitively, $\bar{\bm{Z}}^*$ denotes a re-computation of all the components of $\bar{\bm{Z}}$.

First consider a proposal $Q_u^1$ for the model associated with node $u$.
Recalling that $\bm{M}^*_{[u]}$ is equivalent to $\bm{M}$ with component $M^*_{u}\neq M_{u}$, we have
\begin{align*}
  &Q_u^1\left((\bm{M}, \bar{\bm{Z}}),(\bm{M}^*, \bar{\bm{Z}}^*)\right)\\
  =& \left( q_u^1(M_u,M^*_u) \prod_{v \neq u} \delta_{M_v,M_v^*} \right) \prod_{v\in V}\prod_{M \in \mathcal{M}} \delta_{\widehat{Z}_v^{M}} ((\widehat{Z}_v^{M})^*)\\
  =& \underbrace{\delta_{\bm{M}_{-u},\bm{M}^*_{-u}} q_u^1(M_u,M^*_u)}_{\text{Gibbs-like state proposal}}  \underbrace{\prod_{v\in V}\prod_{M \in \mathcal{M}} \delta_{\widehat{Z}_v^{M}}((\widehat{Z}^M_v)^*)}_{ \text{Same marginal estimates}  }.
\end{align*}
Here $\delta_{x,y}$ is the Kronecker delta, taking value $1$ if $x=y$ and $0$ otherwise and $\delta_z(z^*)$ is, with the obvious abuse of notation, the singular measure concentrated at $z$ evaluated over an infinitesimal neighbourhood of $z^*$. In this proposal, only a change of $M^*_u$ is proposed --- in particular, $\bar{\bm{Z}}$ and $\bm{M}_{-u}$ do not change.
Subsequently, the usual Metropolis-Hastings acceptance probability for this move is:
\begin{align*}
  & 1 \wedge \frac{ \bar{\mu}(\bm{M}^*_{[u]}, \bar{\bm{Z}}) q_u^1(M_u^*,M_u)}{ \bar{\mu}(\bm{M}, \bar{\bm{Z}}) q_u^1(M_u,M_u^*)}\\
  =& 1 \wedge \frac{\widehat{Z}^{M^*_u}_u p(\bm{M}^*_{[u]}) q_u^1(M_u^*,M_u) }{\widehat{Z}^{M_u}_{u} p(\bm{M}) q_u^1(M_u,M_u^*)}.
\end{align*}
For simplicity we have assumed that $q_u^1(M_u,M_u^*)$ is independent of the remaining state variables, but this is not necessary for its correctness and proposals which depend upon their values could easily be implemented.
Note that the augmentation of the state space ensures that $\widehat{Z}^{M^*_u}_{u}$ exists and is available.

Similarly, consider a proposal $Q_u^2$ which refreshes the augmenting variables associated with node $u$:

\begin{align*}
  &Q_u^2\left((\bm{M}, \bar{\bm{Z}}), (\bm{M}^*, \bar{\bm{Z}}^*) \right)\\ =&
                                                                              \delta_{\bm{M},\bm{M}^*}  \prod_{M \in \mathcal{M}} \left( g( (\widehat{Z}_u^M)^* \vert m) \prod_{v \neq u} \delta_{\widehat{Z}_v^m}((\widehat{Z}_{v}^{M})^* ) \right), 
\end{align*}
for which the Metropolis-Hastings acceptance probability is simply
$$  1 \wedge \frac{(\widehat{Z}_u^{M})^*}{\widehat{Z}_u^{M_u}} $$
by exploiting the same cancellation as in the standard pseudo-marginal setting. Clearly moves which update only some of the augmenting variables associated with node $u$ could be justified in the same way.
The augmenting variable $\widehat{Z}^{M}_u$ for any $M\in \mathcal{M}$ can be re-estimated, not be just the normalising constant associated with the current state of $M_u$. This flexibility is the result of the extension of the state space.

Standard arguments allow the combination of moves of these types, and others, within mixtures or cycles to provide an irreducible chain allowing considerable flexibility. In particular, it is no longer necessary to sample a marginal likelihood estimate for every proposed move in the state space.
In the numerical studies below, we evaluate the performance of using combinations of such proposals.

\subsection{Approximations of the NWPM Algorithm}\label{sec:Approximations}
We finish this section with a brief exploration of approximations of the exact pseudo-marginal algorithm described above with the aim of obtaining inference which is almost as good at a fraction of the computational cost by allowing for a small bias in those estimates.

Clearly, a significant portion of the computational load is used to compute the marginal likelihood; this is especially the case when using more sophisticated methods such as the SMC sampler.
For the NWPM method, as presented in Algorithm \ref{algthm: Node-wise PM}, the pseudo-marginal MH Markov chain of length $n$ requires $n \vert V \vert $ marginal likelihood estimates. This quickly becomes infeasible with large data sets such as PET images, which may contain up to $10^6$ voxel time series to be analysed.
It is, therefore, worthwhile to consider strategies that reduce the number of marginal likelihood estimates that are required.


The multiple-augmentation approach to the NWPM algorithm decouples the estimation of marginal likelihoods from moves within the state space. One can envisage, for example, updating the marginal likelihood estimates only every $\kappa$ iterations for some integer $\kappa$. Taking the limiting case as $\kappa\to\infty$ and estimating the marginal likelihood once each model at each node suggests a cheap approximate scheme. Naturally, this approach leads to a Markov chain which is not irreducible on the extended space and which can no-longer be considered a pseudo-marginal algorithm.

Such an algorithm would simply sample each of the marginal likelihood of all the model orders only once, for each node, before starting the chain; then, follow Algorithm \ref{algthm: Node-wise PM} otherwise, using only these initial single estimates. This amounts to making a stochastic approximation which will change the invariant distribution of the resulting Markov chain.

Doing leads to a (now marginal) MH Markov chain, that targets an approximation of the target density rather than the target density itself.
As such, it would not be formally justified under the pseudo--marginal framework, and there would be a loss of accuracy with any subsequent results, since it is an approximation.
Since we only do a single estimation of the marginal likelihoods we will refer to this method as the Node-wise Single-Estimation (NWSE) algorithm.

The biggest appeal of the NWSE is that, for a chain of the same length $n$, the sampler would only need to be used $ \vert V \vert  \vert \mathcal{M} \vert $ times, reducing the costs by $n /  \vert \mathcal{M} \vert $ times.
Any reduction in accuracy and performance could be adjusted for by using some of the saved residual computational resources towards reducing the variance of the sampler marginal likelihood estimates.

This approximation would allow at least preliminary spatial analyses to be conducted with little additional computational cost beyond that required for the associated mass-univariate analysis:
if existing analysis has been done, and thus some estimates of the marginal likelihood have already been obtained then these can be readily used within the algorithm to incorporate spatial dependence.



The NWSE algorithm involves running an MCMC chain with the wrong invariant distribution: that in which the exact marginal likelihoods are replaced with the realisations obtained when they are sampled. It is of intereste, therefore, to establish how different such a distribution is from the true posterior distribution. The following elementary proposition, although it certainly does not provide a formal justification of the NWSE algorithm itself, demonstrates that in the context of small discrete spaces the error introduced by approximating marginal likelihoods can be controlled under reasonable conditions.

\begin{proposition}
  For each model $M \in \mathcal{M}$, let $\widehat{Z}^M:=\widehat{Z}(M)$ be a RV with expectation $Z^M:=f(y\vert M)$, corresponding to the associated marginal likelihood, and variance $\sigma_M^2 < \sigma_{*}^2 < \infty$. If the $(\widehat{Z}^M)_{M\in\mathcal{M}}$ are mutually independent, then, letting $M^* = \arg\max \{Z^M\}_{M\in\mathcal{M}}$ which we assume to be unique the probability of selecting the correct model via maximisation of the marginal likelihood (equivalently, the posterior mode of the distribution over models given a uniform prior over $\mathcal{M}$) is at least:%
\begin{align*}
1 - ( \vert \mathcal{M} \vert -1)\left(1 + \frac{\Delta^2}{2\sigma_*^2}\right)^{-1},
\end{align*}
  where $\Delta = \min_{L\neq M^*} Z^{M^*} - Z^L$ and $\sigma_*^2 = \max_{L\in\mathcal{M}} \sigma_L^2$. 
\end{proposition}
\begin{proof}
  Let $M^* = \arg\max \{Z^M\}_{M\in\mathcal{M}}$, then:
  \begin{align*}
    \mathbb{P}(\{\widehat{Z}^{M^*}& > \widehat{Z}^L \forall L \neq M^*\})\\
    =& 1 - \mathbb{P}(\cup_{L\neq M^*} \{\widehat{Z}^{M^*} \leq \widehat{Z}^L\})\\
    \geq& 1 - \sum_{L \neq M^*} \mathbb{P}(\{\widehat{Z}^{M^*} \leq \widehat{Z}^L\})\\
    \geq& 1 - ( \vert \mathcal{M} \vert -1)\left(1 - \min_{L \neq M^*} \mathbb{P}(\widehat{Z}^{M^*} > \widehat{Z}^L)\right),
  \end{align*}
  by the Union bound.
  Applying Cantelli's inequality, noting that $Z^{M^*} - Z^L > 0$:
  \begin{align*}
    \mathbb{P}(\{&\widehat{Z}^{M^*} > \widehat{Z}^L\})\\ = & \mathbb{P}(\{\widehat{Z}^{M^*} - \widehat{Z}^L - (Z^{M^*} - Z^L) > - (Z^{M^*} - Z^L)\})\\
    \geq & 1 - \frac{\sigma_{M^*}^2 + \sigma_L^2}{\sigma_{M^*}^2 + \sigma_L^2 + (Z^{M^*} - Z^L)^2}\\
    = & 1 - \left(1 + \frac{(Z^{M^*} - Z^L)^2}{\sigma_{M^*}^2 + \sigma_L^2}\right)^{-1}.
  \end{align*}
  Combining these expressions yields:
  \begin{align*}
    \mathbb{P}(&\{\widehat{Z}^{M^*} > \widehat{Z}^L \forall L \neq M^*\})\\ 
    \geq& 1 - ( \vert \mathcal{M} \vert -1) \max_{L\neq M^*} \left[1 -  1 + \left(1 + \frac{(Z^{M^*} - Z^L)^2}{\sigma_{M^*} + \sigma_L^2}\right)^{-1}\right]\\
    \geq& 1 - ( \vert \mathcal{M} \vert -1) \left(1 + \frac{\min_{L\neq M^*} (Z^{M^*} - Z^L)^2}{2\max_{L\in\mathcal{M}} \sigma_L^2}\right)^{-1}\\
    =& 1 - ( \vert \mathcal{M} \vert -1) \left(1 + \frac{\Delta^2}{2\sigma_*^2}\right)^{-1}.
  \end{align*}
\end{proof}

This bound can be made arbitrarily close to 1 by choosing estimators with sufficiently small variance. As the variance of the normalising constant estimates needs to be assessed to allow pseudo-marginal algorithms to be tuned, it is reasonable to suppose that this information could be obtained in settings in which this type of method is used and hence that a reasonable degree of confidence can be obtained that the use of this approximation does not substantially influence the resulting inference. Naturally, this result suggests that if estimating each marginal likelihood once rather than for every algorithmic step as in a pseudo-marginal algorithm, a relatively small variance might be required to obtain good performance.

In the numerical study, presented below, we investigate, evaluate and compare the effects when using the NWSE algorithm in different settings.

\section{Applications}\label{section:EmStudy}
In this section, the NWPM algorithm and some variants are applied to a simple toy model, in which the true marginal likelihood is known, and to a compartmental model for PET data with both simulated and real data.

The software implementation of these methods can be found in the  R package \texttt{bayespetr}, available at: \url{https://github.com/dt448/bayespetr}. The proposed algorithms are somewhat computationally intensive and here we focus on rather short Markov chains, a situation which we expect to be most widely useful; in Appendix~\ref{app:Mcmc-traces-long} we establish that running these chains for much longer does not materially change the resulting inference (as we have found to be the case in all settings which we have explored).

The $20 \times 20$ image displayed in Figure \ref{fig:GTconfig1} will be used to generate the ground truth Potts configuration for all simulated-data experiments. It was adapted from \cite{Bezener2018} and includes a range of spatial structures of varying complexity.
The image is split into four regions, labelled $R_0, \ldots,\linebreak[1] R_3$; all pixels within a given region have a common model order; the details vary between experiments and are given below.

\begin{figure}
  \centering
	\input{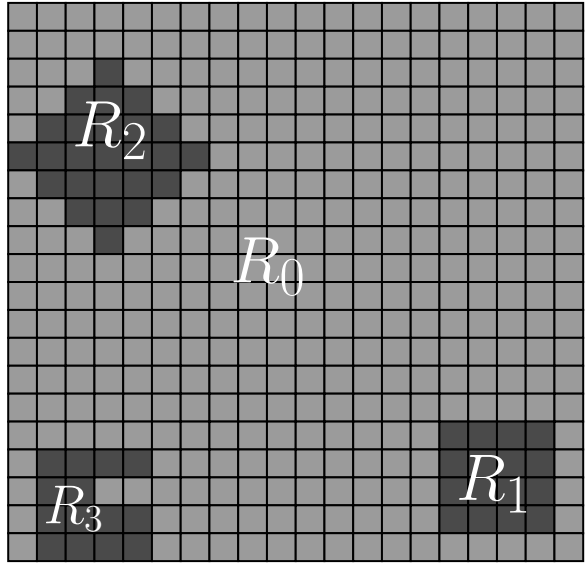}
	\caption{ Ground Truth Configuration: spatial structure used to produce ground truth Potts configuration for the simulated data studies. A $20 \times 20$ grid, with four homogeneous regions.}
	\label{fig:GTconfig1}
\end{figure}

Whenever the NWPM algorithms are used in the numerical studies below, $\bm{M}$ was initialised using the Gibbs a sampler targeting the prior (Potts) distribution.
However, when analysing measured PET data in Section \ref{sec:real-pet-data-analysis}, the algorithms were initialised from the output of the spatially independent SMC sampler method.

In the empirical studies below, we investigated a simple proposal strategy based on the multiple augmentation space, as discussed in Section \ref{sec:Mult-augM-pseudo}.

  In brief, the proposal kernel $Q^1$, together with $Q^2$ at every $\kappa \in \mathbb{N}$ graphical iteration, is used. 
  Essentially, upon initialising the chain, only the change of the model at the node is proposed.
  After every $\kappa-$th graphical iterations, changes in the auxiliary variable or marginal likelihood estimates for every model order at every node is proposed to change.
  In this section, we will refer to this variant as the NWMA algorithm, with tuning parameter $\kappa$.

  \paragraph{SMC Sampler:}
The algorithms considered require unbiased estimates of $\int f(y_v\vert\theta_v,M_v)p(\theta_v\vert M_v)d\theta_v$ at node $v \in V$, for each model order $M \in \mathcal{M}$.
The SMC sampler can be used to compute these node-wise marginal likelihoods by targeting the parameter posterior distribution $\pi(\theta_v\vert y_v,M_v)$ using an annealing schedule.

More precisely, for each model order $M \in \mathcal{M}$, let the sequence of target distribution of the SMC sampler be $\{\pi_t\}_{t \geq 1}$;
Here, we define
$$\pi_t(\theta_{v,t} \vert y_v,M_v) \propto p(\theta_{v,t} \vert M_v,y_v)f(y_v \vert \theta_{v,t},M_v)^{\alpha(t/T)},$$
where the total number of intermediate distributions $T$ and annealing scheme, $\alpha:[0,1]\rightarrow[0,1]$, may be different for each model order \citep[Algorithm 2]{Zhou2016}.
  For example, the annealing scheme could be a simple fixed schedule such as $\alpha(t/T):[0,1]\rightarrow (t/T)^5$, which following \citet{Zhou2016}, we refer to as the ``Prior 5'' annealing scheme.
  Alternatively, $\alpha_t$ can be determined adaptively based on metrics of the distance between intermediate distributions, such as the ESS \citep{Jasra2010,Kong1994} or CESS \citep{Zhou2016}.

For the numerical studies below, we use the SMC sampler.
Pilot studies showed that Prior 5 and the CESS-adaptive annealing scheme produced similar results in these examples, thus the Prior 5 scheme is used for simplicity.

\subsection{Toy Model Simulation Studies} \label{section:NormalDistrToyExample}

\paragraph{Toy Model:}
A simple toy model in which both the prior and model likelihood are normal, at every node, is used to validate the proposed methodology.
The data-set comprises a 2-dimensional digital image, each pixel having a scalar intensity. The graph $G=(V,E)$ used to represent the spatial structure is a finite square lattice, i.e., it has vertex set $V \subset \mathbb{Z}^2$ with nodes $v=(v_1,v_2)^\top$ and
edge set $E=\{ \langle u,v \rangle : d(u,v) = 1 \}$, where $d(u,v) = \sum_i  \vert u_i-v_i \vert $ is the $L^1$ (Manhattan) distance.

In the notation of Section \ref{sssection:S-T model selection with Potts}, given an image $\bm{Y} = (y_v \in \mathbb{R} : v \in V)$, the model studied here is:
\begin{align*}
  \bm{M}  &\sim \text{Potts}(J, G_{\bm{Y}}), &&\\
  \mu_v \vert M_{v} &\sim \mathcal{N}(\mu_{0}^{(M_v)}, \sigma_0^2) \text{ for all }v \in V,&&\\
  Y_{v} \vert M_{v}, \mu_v &\sim \mathcal{N}(\mu_v,\sigma^2) \text{ for all }v \in V;&&
\end{align*}
where
$\mathcal{N}(\mu,\sigma^2)$ denotes the normal distribution of mean $\mu$ and variance $\sigma^2>0$.

At each pixel, $v$,  we model the data point $Y_v$ as a normal random variable with mean $\mu_v$ and variance $\sigma^2$, which is assumed known and is common to all pixels.
The prior over the latent variable $\mu_v$ is normal with mean $\mu_0^{(M_v)}$determined entirely by the model order $M_v$ and fixed variance $\sigma_0^2$.
As before, the model order configuration will be the random variable $\bm{M}$ with a Potts prior distribution.

The node-wise marginal likelihood can be straightforwardly evaluated:
\begin{align*}
  f(y \vert  M) &= \int_{\mathbb{R}} \mathcal{N}(\mu;\mu_0^{(M)},\sigma_0^2)\mathcal{N}(y;\mu,\sigma^2)\dd\mu&\\
          &= \mathcal{N}(y;\mu_0^{(M)},\sigma^2+\sigma_0^2),
\end{align*}
for all  $M \in \mathcal{M}$.
This is trivial to compute since all these terms are known.

The spatial configuration of model order used to simulate the data was specified to represent simple spatial structures; the ground truth value of $\bm{M}$ was fixed to be the Potts configuration shown in Figure~\ref{fig:GTconfig1}. 

\paragraph{Toy Model Ground Truth:}
Ground Truth Configuration, as shown in Figure \ref{fig:GTconfig1} was used to generate a ground truth model order image.
For this toy model, this is a $20 \times 20$ Potts configuration with model order space $\mathcal{M} = \{\mathrm{A},\mathrm{B}\}$.
Region $R_0$ was fixed to model order $M_v = \mathrm{A}$ and the remaining three regions $R_1, R_2$ and $R_3$ where fixed to model order $M_v= \mathrm{B}$.
At each node $v \in V$, the mean parameter $\mu_0^{(M_v)}$ will depend on $M_v$.
In other words, the lighter pixels have nodes with hyper-parameter $\mu^{(A)}_0 = +5$ and the darker pixels have nodes with hyper-parameter $\mu^{(B)}_0 = -5$.
Finally, we specified $\sigma_0^2 = 5^2$ and $\sigma^2 = 1^2$.

\subsubsection{Simulation Study 1: The Coupling Constant}\label{subsection: simulation study 1}
As discussed in Section \ref{sec:potts-model}, typically it is difficult to infer the coupling constant $J$ of the Potts distribution \citep{Moller2006,Moores2020}.
In this work, we make the assumption that $J$ is known when using the NWPM algorithm --- this is not too unreasonable as $J$ is a parameter of the prior distribution.
In this vein, here we evaluate the performance of the algorithm for various values of $J$.

The following design was used: An image data set of $20 \times 20 = 400$ pixels(nodes) from the toy model and model order ground truth as dictated by the Toy Model Ground Truth above, was generated.
This simulated image was then analysed using the NWPM algorithm, for varying values of the coupling constant $J$.
Coupling constants ranging from $J=0$ to $J=5$ were investigated.
Based on pilot studies, an SMC sampler with $N=50$ particles and $T=80$ distributions was used within the NWPM algorithms.
The pseudo-marginal Markov chain was ran for $n=100$ iterations.
As mentioned above, the pseudo-marginal MH Markov chain was initialised using a Gibbs sampler, with the respective coupling constant.
This was done for 50 replicates of the NWPM Markov chain.

Model selection was carried out by selecting the modal model order marginally at each node $v \in V$.
In other words, the modal state in which each $M_v^{(i)}$, for $i=1, \dots, 100$ was selected for each pixel $v \in V$.
The empirical averages, over the $50$ replicates was calculated.
The mean percentages of the nodes for which the correct model order was selected for the different variations of the NWPM algorithm is shown in Figure \ref{graph:Study1State5_2020}.
\begin{figure*}
  \centering
  \begin{subfigure}[b]{0.6\textwidth}
    \includegraphics[width =\textwidth]{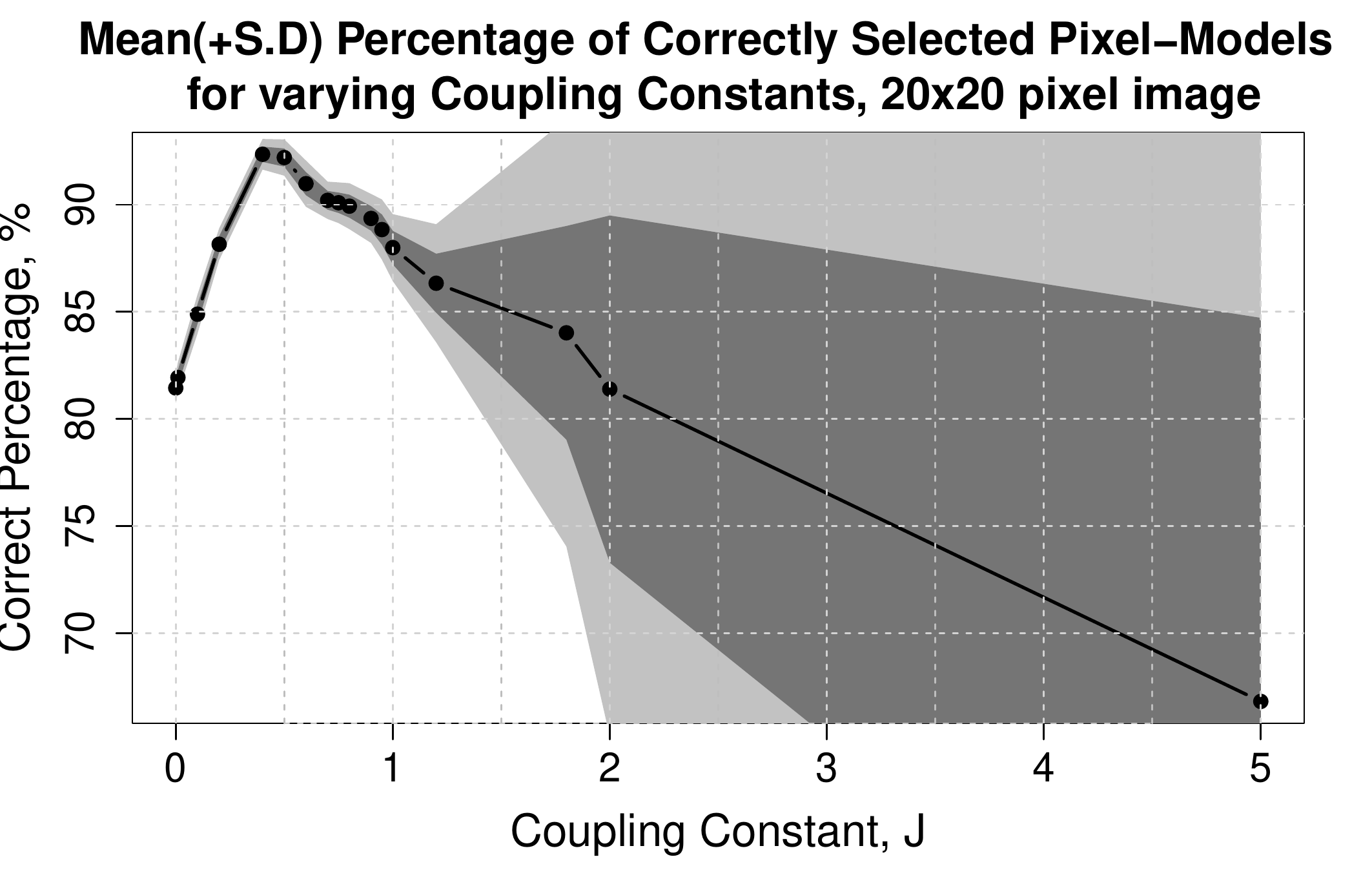}
    \centering
    \caption{Model selection performance for 20$\times$20=400 pixel image.}
    \label{graph:Study1State5_2020}
    \end{subfigure}
    \hfill
    \begin{subfigure}[b]{0.6\textwidth}
      \includegraphics[width =\textwidth]{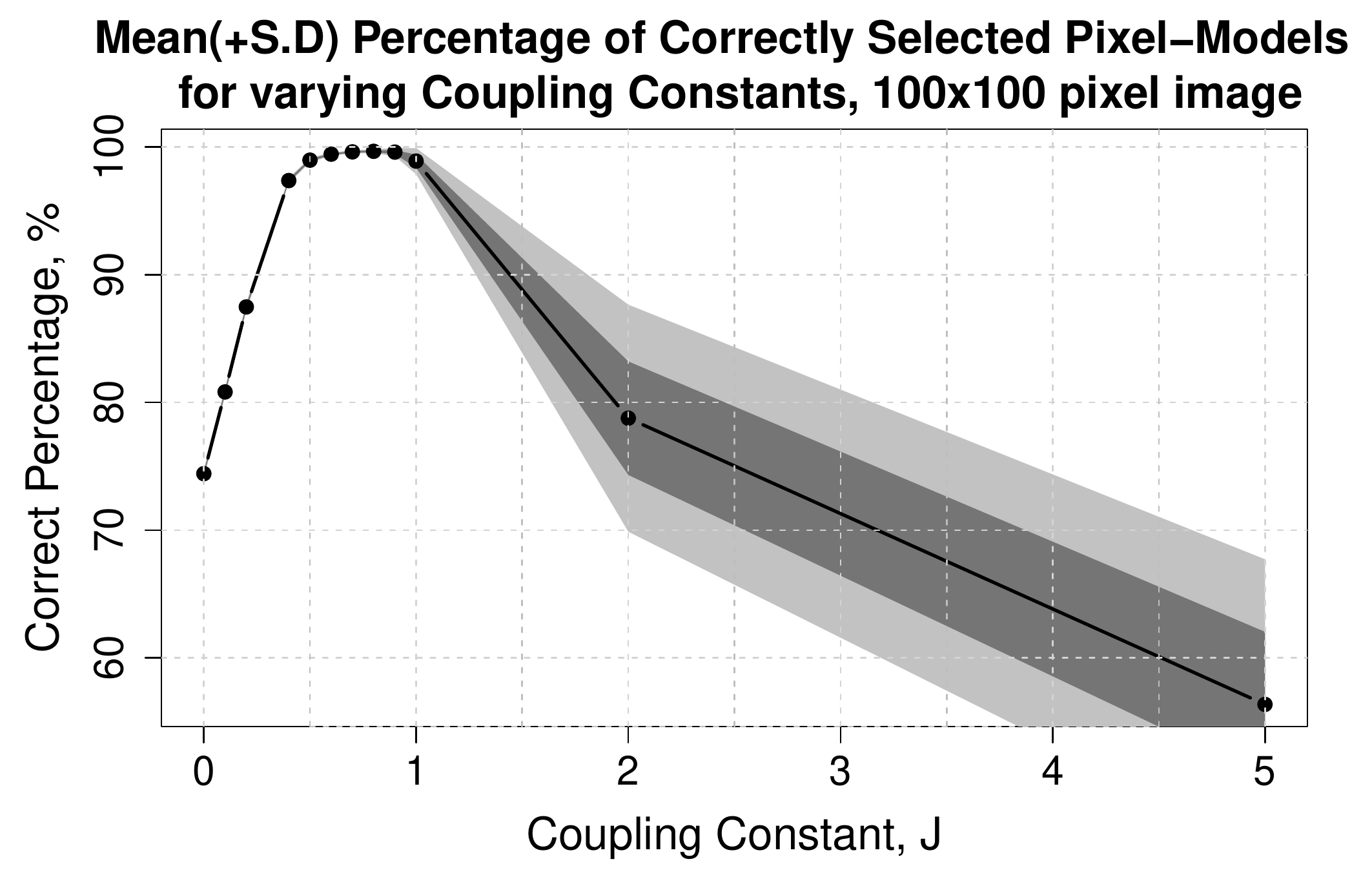}
      \centering
      \caption{Model selection performance for 100$\times$100=10,000 pixel image.}
      \label{graph:Study1State5_100100}
    \end{subfigure}
	\caption{Empirical average, and standard deviations (s.d.), of the percentages(\%) of the total (400 or 10,000) pixels where the correct model order was selected, for different values of the coupling constant $J$. Model selection was done on simulated toy model data using NWPM with an SMC estimator (with $N=50$, and $T =80$) for the marginal likelihood.
    The configuration selected for estimation by looking at the percentage of time each node spent in each model order state.
    Both line plot shows the empirical average of the correct percentage of pixels.
    The lighter and darker region show the $2\times$ s.d. and $1\times$ s.d. error bands, respectively.}
 \label{graph:Study1State5}
\end{figure*}

It is evident from looking at Figure \ref{graph:Study1State5_2020} the performance of the algorithm peaks for coupling constant $J=0.4$.
This is smaller than the critical threshold $J_{\text{critical}}$, as discussed in Section \ref{sec:potts-model}.
This is noteworthy, as there is typically a sharp decrease in speed of mixing of the Markov chain above this critical value. 
In fact, the graph shows that the variance of the percentage of correctly selected model orders increases with the coupling constant; Importantly, there is a jump in the rate of increase between $J=1$ and $J=1.5$.

Given that the size of measured PET data will be considerably bigger, this experiment was repeated for a larger Toy Model image.
Additionally, the model order space for compartmental models may be larger; For instance, in this paper we consider up to $M=3-$compartment models. 
Subsequently, a scaled version of the 20$\times$20 image, with dimensions 100$\times$100 =  10,00 pixels and model space $\mathcal{M} = \{\rm{C},\rm{D},\rm{E}\}$.
More specifically, the region corresponding to $R_0$ had mean parameter $\mu_0^{(\rm{C})} = +7$; Similarly, region $R_1$ had  $\mu_0^{(\rm{D})} = 0$ and regions $R_2$ and $R_3$ had $\mu_0^{(\rm{E})} = -7$.
The parameters $\sigma_0^2$ and $\sigma$ were kept as above.
The generated data was the analysed using the same the methods and design.
The results are shown in Figure \ref{graph:Study1State5_100100}.

Notably, we first see that the peak of the graph is flatter in comparison to Figure \ref{graph:Study1State5_2020}.
Secondly, the performance of algorithm is better overall; For instance, the best performance is almost $100\%$ compared to roughly $94\%$ previously.
Similarly, when looking at the performance corresponding to $J=0.4$, it increases from $94\%$ to $97\%$. 
This suggest that, in this context, a larger image size and model order space did not seem to impact the performance negatively.
Similarly, there is evidence that a wider range of coupling constant values could be used in the setting with bigger images and larger model order space, to attain ``optimal'' model selection performance.
A selection of model order outputs, for varying values of $J$, are given in Figure \ref{fig:thre_model_toy}.

When looking at these results, as presented in Figure \ref{graph:Study1State5}, an important question to consider here is how much of the performance difference is due to poor approximation of the posterior by the Markov chain, and how much to the posterior itself.
It is reasonable to argue that below the critical value, reasonably good mixing is obtained and this probably reflects a good approximation of the posterior as indicated by the low variability.
In contrast, above this critical value the results are largely dominated by the poor mixing of the Markov chain over this rather short period.
We also note that the graph shows good performance for even this relatively short chain, for the appropriate coupling constant values.

Based on evaluating the performance for a range of coupling constant and due to the similar qualitative features of the different settings, in the interest of simplicity, for all the following studies we will fix $J=0.4$. This lower value means that the risk of being too close to the critical value $J_{\rm{critical}}$ is avoided.

\subsubsection{Simulation Study 2: Sample Size vs Chain length} \label{sebsection: simulation study 2}

\begin{figure*}
  \centering
  \begin{subfigure}[b]{\textwidth}
	\includegraphics[width = \textwidth]{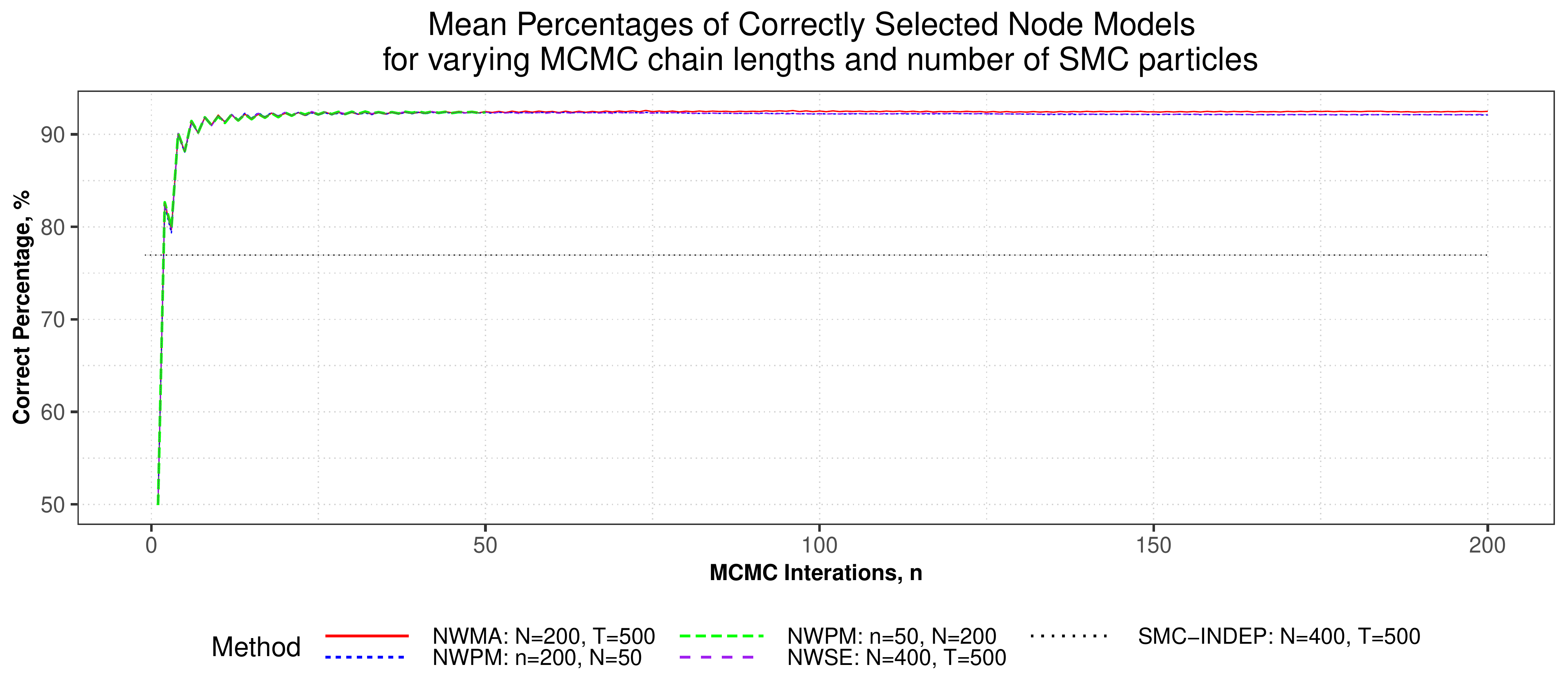} \centering
  \caption{Mean percentages at each graphical (MCMC) iteration.}
  \label{graph:Study2State5_iter}
  \end{subfigure}
  \hfill
  \begin{subfigure}[b]{\textwidth}
    \includegraphics[width = \textwidth]{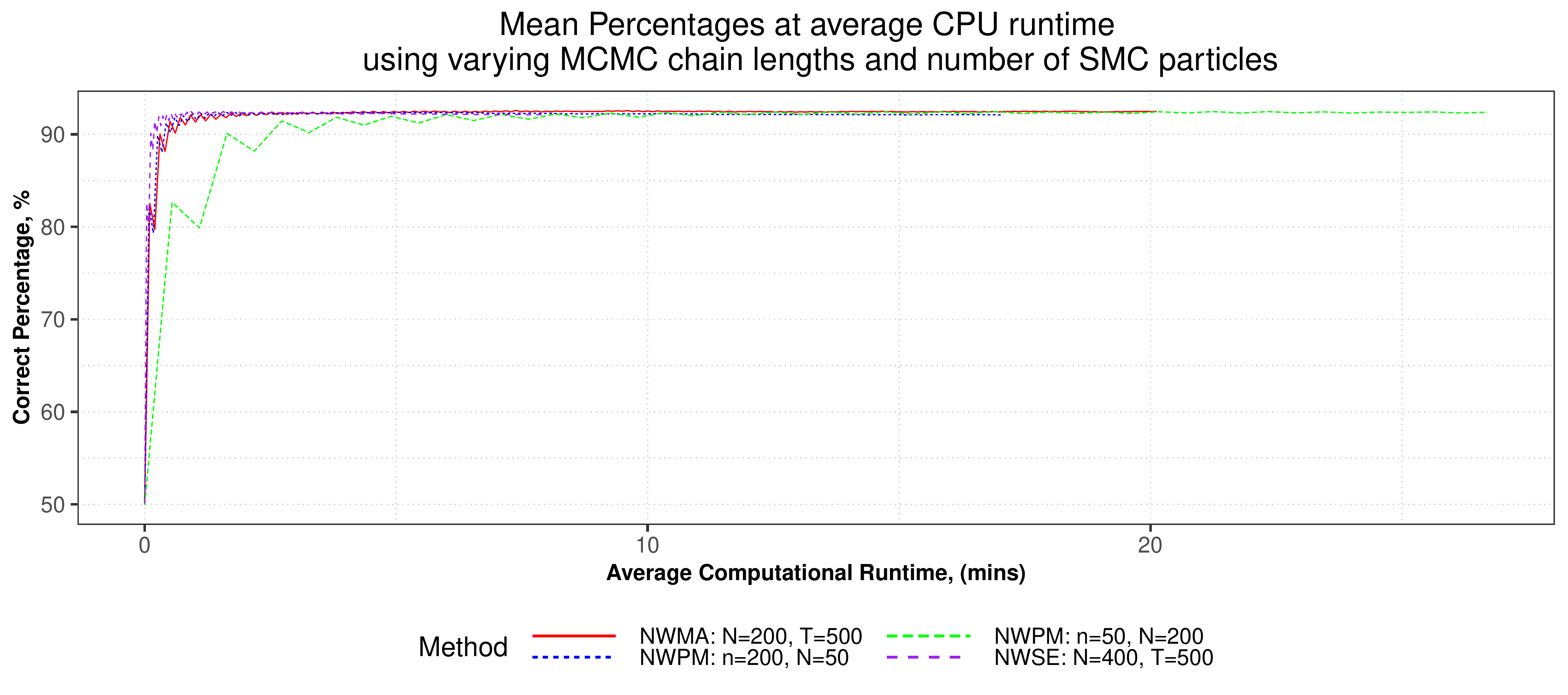} \centering
    \caption{Mean computational runtime (averaged over each graphical iteration).}
    \label{graph:Study2State5_comp}
  \end{subfigure}
    \caption{Average percentages (\%) of the whole image ($20 \times 20 = 400$ pixels, Toy Model) where the correct model order was selected at each iteration (shown in part (a)) of the MCMC chain, using NWPM, NWSE and NWMA for varying Markov Chain length, $n,$ and number of particles, $N,$ in the SMC sampler. The dashed line shows the average percentage when using spatially independent SMC model selection.
    The cumulative average over the graphical iterations $i=1, \dots, n$ is shown.
    Average computational runtimes are shown is shown part (b).
}
\label{graph:Study2State5}
\end{figure*}

Next, the computational trade--off between the number of particles used in the SMC sampler and the length of the pseudo-marginal MH Markov Chain was investigated.
The same design as the study in Section \ref{subsection: simulation study 1} was used, with the exception of the Markov Chain length varying from $n=50$ to $n=200$, and the number of particles varying between $N=50$ and $N=200$.
Here, we investigate the 20$\times$20=400 pixel image.

Additionally, spatially independent model selection, using an SMC sampler to estimate the model evidence was also investigated.
The NWSE approximation  was also used to analyse the image.
Similarly, the NWMA algorithm, as described above, with updates to the marginal likelihood estimates every $\kappa=10$ graphical iterations was used.
For this variant, an SMC sampler with $N=200,T=500$ was used.
For these two, less resource--intensive, Markov chain methods (NWMA and NWSE), $n=200$ graphical iterations were used.

Model selection was carried out in the same manner as described above (i.e. the modal model), and the percentage of correctly selected nodes was used as the metric of performance.
This was done for 100 replicates, for each method used.
Results are shown in Figure \ref{graph:Study2State5}.

These results suggest that, at least in this simple setting, the algorithm is not particularly sensitive to the allocation of computational resources.
Importantly, we can see that the selected graph model/configuration stabilises fairly quickly for all the combination of tuning parameters. 
In particular, looking at the longer $n=200$ chain in Figure \ref{graph:Study2State5_iter}, it is evident that the convergence is to around $92\%$ of nodes having the correct model order selected when using the NWPM algorithm.
This is considerably higher than the performance of the spatially independent SMC sampler, which was $79\%$; confirming that exploiting spatial structure can significantly improve estimation in this type of model.
We also observe that the NWSE approximation method and the simple NWMA algorithm achieved almost identical results to the full NWPM algorithm: suggesting that, in this simple setting, there is enough signal to allow for the use of the approximation to achieve very similar results, with very little additional computational cost.
Intermediate chain lengths not shown in the graph, i.e. $n=75$ and $N=134$, also produced similar results.

Figure \ref{graph:Study2State5_comp} shows the average computational runtime for each method. 
The runtimes were computed by considering the average CPU time per iteration:
That is, the total run time for each trail was recorded and then averaged over the total graphical iteration.
We can draw similar conclusions as above, with the exception that the NWPM method with $N$ = 200 particles seemed to take a longer time (in minutes) to reach stability.

The computational runtime for the NWMA method, for varying values of $\kappa$, was also studied: Results are shown in Figure \ref{graph:NWMA_kappa_comp}. As expected, the total runtime increase for smaller values of $\kappa$; However, a higher refresh rate seems to slightly slow down (in minutes) how quickly the method attains its stable state.    

Finally, the long term behaviour for the NWSE and NWMA MCMC chains was also explored and results (i.e. that much longer chains lead to very similar conclusions) are shown in Appendix \ref{app:Mcmc-traces-long}, Figure \ref{graph:long_run_toy}. 
Additionally, Figure \ref{graph:Study2State5_zoom} shows an enhanced version of Figure \ref{graph:Study2State5} --- allowing for better discernment of the small differences in performance. 

\subsection{Simulated PET Data Studies}

Before we investigate the performance of the proposed algorithms on measured PET data, we verify and investigate tuning factors when applied to simulated PET data.

\subsubsection{Preliminaries}

Consider a time series from a single voxel site of a PET image, denoted $y=(y_1, \dots, y_k)^{\top}$.
Recall the $M-$compartmental model for PET data (see \eqref{eq:comp_signal} and \eqref{eq:comp_model} , Section \ref{PET}).
In particular, note that the latent parameters $\phi_{1:M} = (\phi_1, \dots, \phi_M)$ and $\vartheta_{1:M} = (\vartheta_1, \dots, \vartheta_M)$ determined the dynamics of this model.

For the PET simulation study in this section, we simulate data from this collection of models with normally distributed errors as described in Section \ref{PET}.
For simplicity, we simulate data using a common parameter vector for every voxel associated with a given model order.
These parameter values were taken from \cite{Jones1994} (and have subsequently been used for simulation studies by \cite{Gunn2002,Peng2008,Jiang2009} and \cite{Zhou2013}). 
For the $M=1$ and $M=2$ compartmental models we use $\phi_1 = 4.9 \times 10^{-3} s^{-1}$, $\phi_2= 1.8 \times 10^{-3}s^{-1}$, $\vartheta_1 = 5 \times 10^{-4} s^{-1}$ and $\vartheta_2 = 0.011 s^{-1}$; For the $M=3$--compartmental model we use $\phi_{1:3} = (4.4 \times 10^{-3}s^{-1}, 1 \times 10^{-4}s^{-1}, 1.4 \times 10^{-3}s^{-1})$ and $\vartheta_{1:3} = (4.5 \times 10^{-4}s^{-1}, 2.7\times10^{-3}s^{-1}, 1 \times 10^{-2}s^{-1})$.
In each of these cases the volume of distributions is roughly $V_{\mathrm{D}} \approx 10$.

Synthetic data was simulated using the same plasma input function and integration periods as in the measured PET data set studied in Section \ref{sec:real-pet-data-analysis} and normally-distributed noise, with zero mean, was added.
The variance was such that it was proportional to the true time activities divided by the length of the frame duration; and scaled such that the highest variance in the sequence is equal to the noise level.
A noise level of $0.5$, which lies within the representative range (between $0.01$ and $1.28$) of real data \citep{Jiang2009}, was used when simulating the 2-D PET image.

We employed vague uniform priors over $\phi$ and $\vartheta$, which were constrained to lie within $[10^{-5}, 10^{-1}]$ and $[10^{-4},10^{-1}]$, respectively.
These ranges are based on pilot and previous studies, and ensure that the parameters are physiologically meaningful \citep{Cunningham1993}.

The precision parameter of the measure noise is denoted $\lambda = 1/\sigma^2$. A gamma distribution, with both parameters equal to $10^{-3}$, was used as the prior for $\lambda$ -- a proper approximation to the improper Jeffrey's prior.

Details of the relevant distributions can be found in Appendix~\ref{App:Model_Equations}.

\paragraph{PET Simulation Ground Truth}
As before, the ground truth of the model order configuration were based on spatial structure pattern of Ground Truth Configuration, shown in Figure \ref{fig:GTconfig1}.
In this PET simulation study, a $20 \times 20$ Potts configuration with model order space $\mathcal{M} = \{1,2,3\}$, representing the number of compartments was generated.
Region $R_0$ corresponds to model order $2$; region $R_2$ and $R_3$ to model order $1$ and region $R_1$ to model order $3$.

\subsubsection{Simulation Study 1: Algorithm Comparison}

Pilot studies were initially carried out to evaluate the model selection performance and estimator variance of the SMC sampler.
For the variance of the estimates, on agreement with numerical studies by \cite{Zhou2016}: the numerical study showed that, there was a decrease linearly for higher values of $N$ and similarly in $T$.
However, beyond a certain threshold, the reduction were no longer useful.
Similar results were seen in pilot studies for model selection performance, for 1-D data simulated from just the $2-$compartment model with noise level $0.5$, performance could not be increased beyond $65\%$ regardless of any increase in tuning parameter values.

Next, the simulated 2-D PET image was analysed using the NWPM algorithm, and the effect of the pseudo-marginal MH chain was investigated.
Based on the above pilot study, an SMC sampler with the Prior 5 annealing scheme, $T=400$ intermediate distributions, was used as the estimator of the normalising constants at all pixels for all models.
These values, and the values of the number of particles $N$ in the design below, were chosen as they allow for the variance of the likelihood estimates to be within a suitable range to allow for optimal scaling \citep{Doucet2015}.
Chain lengths of $n=50$, $n=75$, $n=100$ and $n=200$ were investigated, with particles $N=200$, $N=134$, $N=100$ and $N=50$, respectively.

These, relatively smaller, graphical iteration lengths were chosen due to the computational overhead of this algorithm in this resource intensive setting.
Note that, since model selection here uses the marginal modal state from a space of just $3$ model orders --- these shorter chains can still produce reliable results. 

Additionally, the image was also analysed using spatially independent SMC sampler and the NWSE approximation; using an SMC sampler with $N=400$ particles and $T=600$.
The NWMA algorithm, with SMC sampler using $N=200$ and $T=600$, at every $\kappa=50$ iterations, was also evaluated.  
Given their significantly reduced computational costs, for the two pseudo-marginal methods, $n=500$ iterations were used.

Following Algorithm \ref{algthm: Node-wise PM}, all the chains where initialised using the Gibbs sampler to target the prior distribution.
$30$ independent replicates of each algorithm were used to obtain the results shown in Figure \ref{graph:PETSimStudy} and Table \ref{tab:1} (performance for the intermediate ranges of tuning parameters of the NWPM algorithm is shown in Appendix \ref{app:Mcmc-traces-long}, Figure \ref{graph:PETSimStudyInterRange}).

\begin{figure*}[h!]
  \centering
  \begin{subfigure}[b]{\textwidth}
	\includegraphics[width = \textwidth]{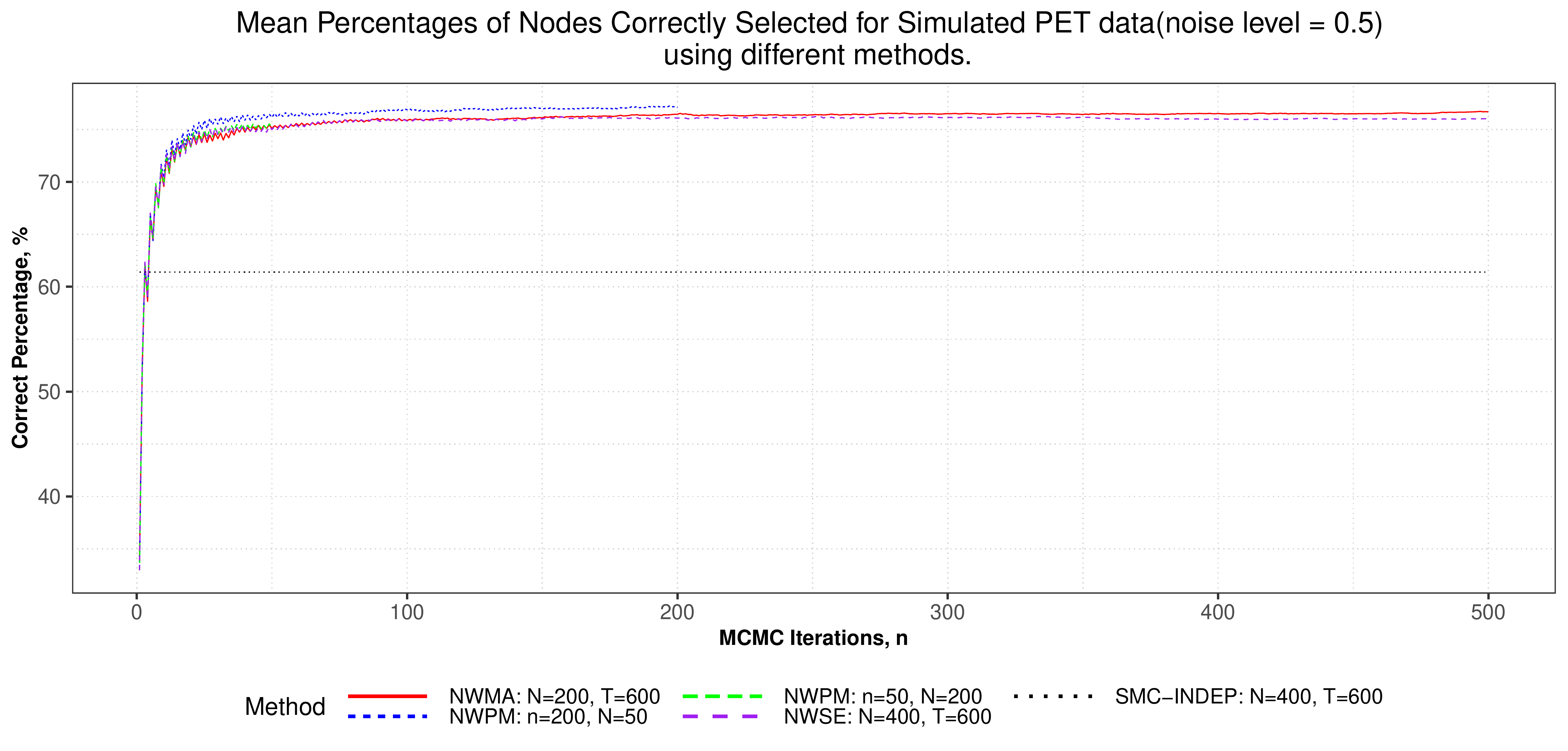}
  \caption{Mean percentages at each graphical (MCMC) iteration.}
\label{graph:PETSimStudy_iter}
  \end{subfigure}
  \begin{subfigure}[b]{\textwidth}
    \includegraphics[width = \textwidth]{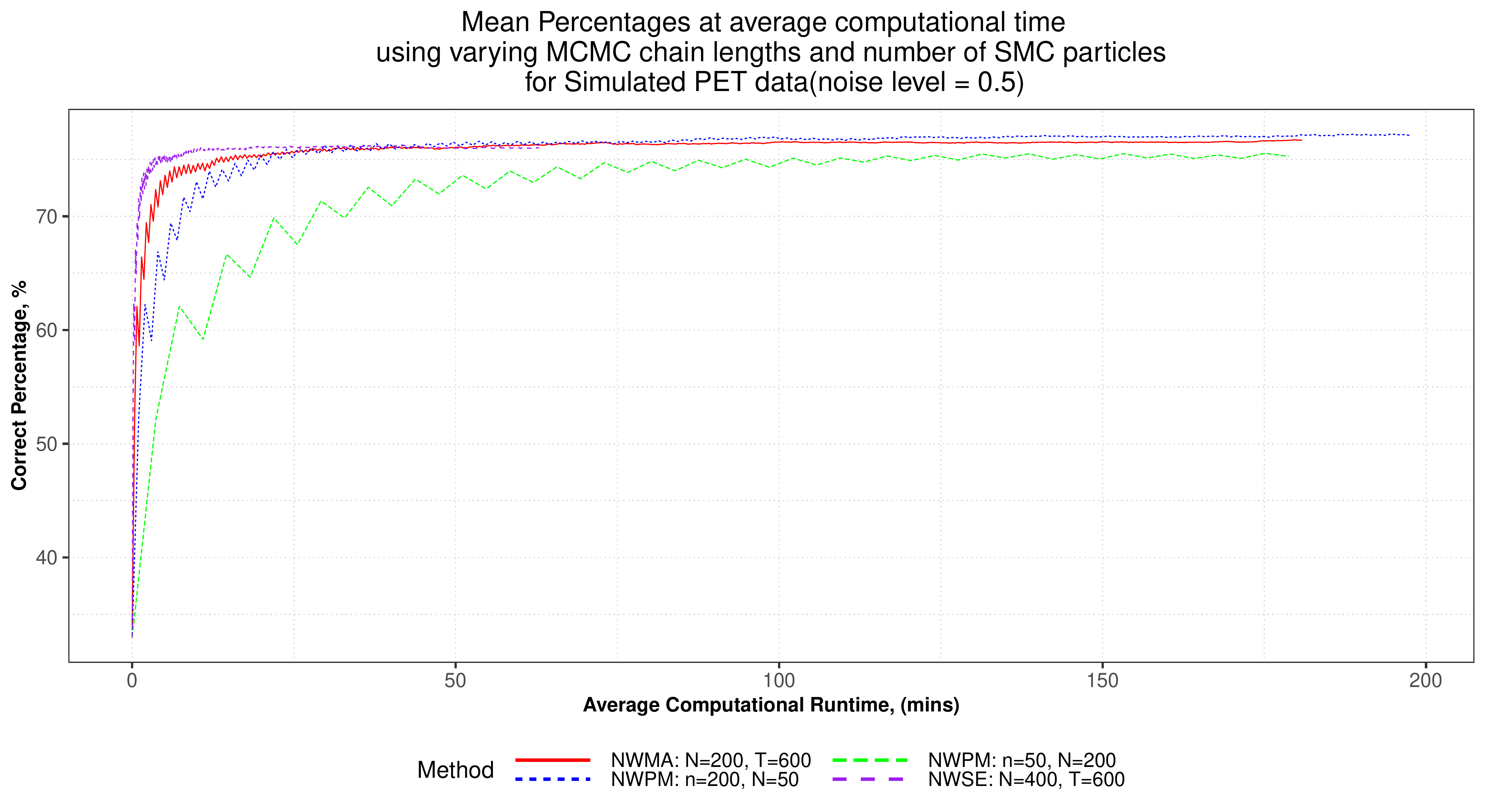}
    \caption{Mean computational runtime (averaged over each graphical iteration).}
\label{graph:PETSimStudy_comp}
  \end{subfigure}
\caption{Average percentages(\%) of the whole image ($20 \times 20 = 400$ pixels, simulated PET data) where the correct model order was selected at each iteration of the pseudo-marginal MH chain, using: NWPM for varying Markov Chain length, $n,$ and number of particles, $N,$ in the SMC sampler; NWSE with $N=400$ particles and $T=600$ distributions and NWMA using SMC sampler with $N=200$ and $T=600$ at every $\kappa=50$. The dashed line shows the average percentage when using spatially independent SMC($N=400,T=600$) model selection.}
  \label{graph:PETSimStudy}
\end{figure*}

Looking at Figure \ref{graph:PETSimStudy_iter},  the results show notable improvements in model selection performance when using the different variants of the NWPM algorithms compared to spatially independent model selection using the SMC sampler.
Analysing the simulated 2-D PET image using the spatially independent method gave $61.4\%$ pixels correctly selected, compared to the NWPM algorithms all giving around $76\%$, even for as short as $n=50$ graphical iterations.
On average, this is a $23.4\%$ increase in the number of pixels selected correctly when using the proposed NWPM approaches compared to the spatially independent SMC sampler method.
On large data-sets, like PET images, this is a considerable number of pixels.
Furthermore, since this is spatially dependent model selection, the improvement in correctly selected pixels when compared to SMC independent will be towards revealing the underlying spatial structures.
These results suggest that incorporating spatial dependence does results in better model selection performance.

Figure \ref{graph:PETSimStudy} also shows that the performance of the NWSE approximation and the NWMA variant both provide similar performance to the full NWPM algorithm.
Pragmatically, this is very promising as it suggests that even with an approximation current state of the art methods can be outperformed with very little additional costs.
Alternatively, rather than using an approximation, algorithms based on the multiple augmented state space with significantly reduced computational overhead can be used.
For example, the simple NWMA algorithm studied here. 
Similar to the toy model, there does not seem to be any notable difference when using higher chain length and lower particles.
Additionally, there is only a small increase in performance when using more a computationally intensive algorithm, on average.
That is, the full NWPM performs slightly better than the NWMA variant; which in turn also has a minute increase in performance compared to the NWSE approximation.

Additionally, for the chains where higher iterations were used, we can see that this increase in results was maintained.
To verify long term behaviour of the algorithms: long-run MCMC chains, of the less expensive NWSE and NWMA algorithms, are shown in Appendix \ref{app:Mcmc-traces-long}, Figure \ref{graph:long_run_PET}. 
As before, an enhanced version of Figure \ref{graph:PETSimStudy_iter} is shown in Figure \ref{graph:PETSimStudy_zoom}.

As before, Figure \ref{graph:PETSimStudy_comp} shows the average computational runtime for each method. The runtimes
were computed in the same manner as described in Section \ref{sebsection: simulation study 2}. Similarly, as seen in the aforementioned section, the NWPM method with n = 200 seems to be the slowest in regards to
mixing (in minutes).

Next, Table \ref{tab:1} shows the average mean squared error (MSE) of the $V_\text{D}$ estimates for each variant of computational method.
The quantity was calculated by taking the average MSE over all the pixels in the 2-D PET image, and then over all the replicates (the empirical variance of these quantities are shown in the parenthesis).
Although the non-spatial SMC sampler method has the highest empirical MSE,
the NWPM algorithms, including all chain lengths and the SE variant, produce a noticeably smaller MSE compared to the spatially independent SMC method, and
there is a small improvement when using the full pseudo-marginal NWPM, compared to its approximation, these differences all fall within the range that could be plausibly explained by random fluctuations and it is difficult to draw any meaningful conclusion.

\begin{table}[h!]
  \centering
  \caption{Average RMSE (and s.e.) of $V_D$ estimates for simulated PET data.}
  \label{tab:1}       
  \begin{tabular}{lll}
    \hline\noalign{\smallskip}
    Method & Model Averaging $V_\text{D}$ & Posterior $V_\text{D}$  \\
    \noalign{\smallskip}\hline\noalign{\smallskip}
    NWPM $n=$50 & 0.1319 (0.0061) & 0.1321 (0.0061) \\
    NWPM  $n=$75 & 0.1321 (0.0061) & 0.1323 (0.0061) \\
    NWPM $n=$100 & 0.1320 (0.0061) & 0.1323 (0.0061) \\
    NWMA & 0.1319 (0.0061) & 0.1323 (0.0061) \\
    NWSE & 0.1322 (0.0061) & 0.1324 (0.0061) \\
    INDEP-SMC & 0.1341 (0.0066) & 0.1341 (0.0063) \\
    \noalign{\smallskip}\hline
  \end{tabular}
\end{table}

\subsection{Measured [$^{11}$C]diprenorphine PET image Study} \label{sec:real-pet-data-analysis}
The simulation studies above have verified the effectiveness of the proposed methodology and allowed for the identification of tuning parameters which lead to promising results in settings close to that of interest. We now turn our attention to real data sets.

Analysis of the tracer kinetics of dynamic PET images, that use the [$^{11}$C]-diprenorphine tracer, are generally aimed towards quantification of opioid receptor concentrations in the brain. 
Such studies involve investigation of neurological disease which often involve change in brain receptor density.

Investigation of similar PET data sets, using different statistical approaches, have been previously analysed quantitatively by \cite{Gunn2002}, \cite{Peng2008} and \cite{Jiang2009}; However, these works focused on parameter estimation rather than model selection.
\cite{Zhou2013,Zhou2016} are more recent works that investigated model selection; Albeit, in these works, spatial independence is assumed.  

\subsubsection{PET Data}

A test-retest pair of measured dynamic PET data of the concentration of the tracer [$\vphantom{C}^{11}$C]–diprenorphine in a normal subject, for which an arterial input functions were available, was analysed using the proposed and other methods. This data was initially measured as part of the dynamic PET study in \citet{Hammers2007}.
In that study, the effects of using lower concentration of tracer was investigated. In addition, the impact of using different slow frequency cutoffs within the analysis of the data was also studied.

\paragraph{PET data acquisition}
The subject underwent 95 min dynamic [$^{11}$C]–diprenorphine PET baseline scan.
[$^{11}$C]–dipre\linebreak[1]norphine is a tracer, of the non-selective antagonist, which binds to neural opioid receptors.
The PET scans were acquired in 3D mode on a Siemens/CTI ECAT EXACT3D PET camera.
After image reconstruction, the spatial resolution is approximately 5mm.
Technical details including the data correction and normalisation methods of the ECAT EXACT3D PET camera can be found in \citet{Spinks2000}.
Thirty seconds after the start of the scan, [$^{11}$C]-diprenorphine was injected and flushed through the cannula as a smooth bolus over a total of 30s.
The subject was injected with 185 MBq (Mega Becquerel) of [$^{11}$C]–diprenorphine.
The PET data was reconstructed using the re-projection algorithm \citep{Kinahan1989} with ramp and Colsher filters cutoff at the Nyquist frequency.
Reconstructed voxel size were $2.096\text{mm} \times 2.096\text{mm} \times 2.43\text{mm}$.
Acquisition was performed in listmode (event-by-event) and scans were rebinned into 32 time frames of increasing duration.
The lengths of these periods, in seconds, are: (variable length background time, $3 \times 10$, $7 \times 30$, $12 \times 120$, $6 \times 300$, $75$, $11 \times 120$, $210$, $5 \times 300$, $450$ \text{ and } $2 \times 600$).
Frame-by-frame movement correction was performed on the PET images.
Overall this resulted in images of dimensions/size $128 \times 128 \times 95$ voxels.
This gives a total of $1, 556, 480$ separate times series respectively to be analysed.
In this study, we will analyse a cross-section of each of the sagittal, coronal and transversal planes of the brain --- the total number of times series studied here is, therefore, roughly 11,000.

Details of further PET data pre-processing can be found in \citet[Material and methods]{Hammers2007}.

\paragraph{Decay Correction}
The PET data is not decay corrected and this should be accounted for within the compartmental model.
\citet{Gunn2002} (similarly, \cite{Cunningham1993} for Spectral Analysis), suggest that the slow frequency cutoff be close to the decay constant of the tracer used.
In other words, given that the half life of [$^{11}$C] is 20.4min, its decay constant is $0.0005663$s$^{-1}$;
Thus, the slow frequency cutoff, denoted $\theta_{\text{min}}$, is fixed close to this constant.
\cite{Hammers2007} showed that actual parametric values depended heavily on the cutoff slow frequencies( between 0.0008 s$^{-1}$ and  0.00063 s$^{-1}$). 
Importantly, the prior ranges of $\theta_{1:M}$ should be based upon $\theta_{\text{min}}$, and subsequently this should be taken into consideration when computing the volume of distribution $V_{\mathrm{D}}$.
Following \cite{Zhou2016}, in particular see  \cite{Zhou2014}, we will use the cutoff $0.0007s^{-1}$.

\subsubsection{Data analysis and performance evaluation}\label{sec:data-analys} 

Normally-distributed errors have been found to poorly model data of this sort \citep{Zhou2013}.
Instead we will use the $t-$distribution to model the additive error variable $\epsilon$.
The same gamma prior as used for $\lambda$ in the simulation studies above, will be used for the scale parameter, $\tau$. A uniform distribution over the interval $[0,0.5)$ will be used as the prior for $1/\nu$, allowing the likelihood to vary from having a very heavy tail to being arbitrarily close to normality.

The measured data was analysed using the following methods:
Firstly, SMC sampler estimates of the marginal likelihood of each pixel for the three models were computed.
This strategy assumes spatial independence, so, as before, we will refer to this method as the independent SMC method.
The Prior 5 annealing scheme with $N=300$ particles and $T=500$ distributions was used, these values were based on the studies above and numerical results reported by \citet{Zhou2016}.
Three cross-sectional slices (coronal, transverse and sagittal planes) of the 3-D PET image were analysed, the model order and volume of distribution parametric images are shown in Figure \ref{fig:REAL_PET_SMC}.

Next, the NWPM algorithm was used to analyse the same three cross-sections.
An SMC sampler with $N=300$ particles and $T=500$ distributions was used as the marginal likelihood estimator.
Here, the chain was initialised from the output of the SMC spatially independent model selection above (rather than from the prior distribution, as in the simulation studies above).
A pseudo-marginal MH chain of length $n=75$ was generated.
This chain length was based on viability of computation cost for this relatively large data set, and the simulation studies above.
In addition, the sagittal cross-sectional image was analysed using  a longer $n=200$ iterations, as discussed below.
The results are shown in Figure \ref{fig:REAL_PET_NWPM}.

Finally, Figure \ref{fig:REAL_PET_NWPM_SE} and  \ref{fig:REAL_PET_NWPM_MA} show the model configuration output from analysing the measured PET data using the NWSE approximation and the NWMA variant, respectively.
In this case, an SMC sampler with $N=400$ particles and $T=600$ distributions was used to make the marginal likelihood estimates.
The pseudo-marginal chain was initialised using the output of these estimates, and these initial estimates where re-used within the NWSE algorithm for $n=500$ iterations.
For the NWMA method, the chain was initialised in the exact same manner, but the marginal estimates where updated every $\kappa=100$, for a total of $n=500$ iterations.

For further comparison, the sagittal cross-section was also analysed using  NLS (non-linear least squares; \cite{Zhou2013}).
This is shown in Figure \ref{fig:REALPET_ALGO_COMP}.
The spatially independent SMC sample model selection and NWPM method(with a longer chain length $n=200$) model selection are included to show the difference in the final model order parametric image output.

The volume of distribution images, specifically Figure \ref{fig:REAL_PET_SMC} and Figure \ref{fig:REAL_PET_NWPM}, show very similar performance for all methods when considering inference of $V_\text{D}$.
However, Figure \ref{fig:REALPET_ALGO_COMP} shows clear difference in the model order image when using NWPM compared to spatially independent SMC and NLS.
It is evident that analysis using NWPM reveals more spatial structure: Albeit, using the SMC independent approach does reveals some of the underlying spatial information when compared to NLS; NWPM has a de-noising effect and improves the clarity of the roughly formed clusters seen in the spatially independent SMC method.
In particular, the $M=1$ and $M=2$ compartment models are selected to model most of the structural clusters.
In contrast to this there is a decrease in the number of $M=3$ compartment models selected when compared to the non-spatial SMC method.

Note also that, by looking at Figure \ref{fig:REALPET_ALGO_COMP} we see that the NWPM algorithm converges to a noisy version of the final output image even at low iterations as $n=25$.
Importantly, there seem to be only small changes at the higher iterations, suggesting that the chain reaches a stable output relatively quickly.
Similar behaviour is seen in the simulation studies above, albeit in a smaller, simpler settings.

Furthermore, similar model selection output is seen when using the SE estimate approximation and the multiple augmentation variant, NWMA, of the NWPM algorithm.
Comparing both Figure \ref{fig:REAL_PET_NWPM_SE} and \ref{fig:REAL_PET_NWPM_MA} with Figure \ref{fig:REAL_PET_NWPM}, we see that the mode order outputs are not identical, but qualitatively show many similarities.
In particular, we see that the there is a decrease in the total number of $3$-compartment models selected when using the SE approximation or the NWMA variant, compared to the full NWPM. 
A similar effect is seen when comparing the independent SMC model selection with the NWPM.
Notably, the NWSE and NWMA model selection outputs are almost identical, differing only at a very small number of pixels.

\section{Conclusion}

In this work we have illustrated a novel computational method for effectively incorporating spatial dependence when performing model selection at each location within a graph.
This approach extends the pseudo-marginal method, in a number of directions, within the context that it has been developed and evaluated in.
By exploiting the structure of the problem at hand, the developed methodology allows for considerable flexibility in the updating of state variables and marginal likelihood estimates relative to generic pseudo-marginal algorithms.
The empirical study suggests that this approach can yield better inference than methods which impose assumptions of full spatial independence and, in particular, can reveal clearer spatial structure in the image of underlying model orders.
As with pseudo-marginal algorithms, essentially any unbiased marginal likelihood estimators used in non-spatial analysis can be used within this algorithm.

This methodology was motivated by considering the problem of model selection in PET images.
In such a setting both the amount of data and the complexity of the models lead to substantial computational requirements.
The images of model order provided by this approach supplement the information provided by volume of distribution and similar macroparameters in PET images.
For instance, Figure \ref{fig:REALPET_ALGO_COMP} (part d.) could be used to make the interpretation that the signal from some parts of the PET image require a higher model order compared to other parts .
Specifically, this is apparent when comparing the posterior region of the brain image (where 2-compartmental models were selected) to the anterior region of the brain (where 1-compartmental models were selected).

In order to reduce the computational cost, a multiple augmentation variant of the pseudo-marginal algorithm is introduced and a more extreme approximate form in which each marginal likelihood is estimated only once. The performance of these methods in the examples explored here is very encouraging:
it suggests that where non-spatial analysis has been performed, spatial dependence can be incorporated with the existing estimates very easily and with little computational costs.
Doing so carefully can result in similar performance to the full NWPM algorithm.

In summary, the NWPM algorithm, together with its approximations and extensions, enables a intuitive, natural approach to incorporating spatial dependence in realistic and complex settings; such as analysis of PET images.
In particular, the methods extensions and approximations means that the subsequent performance and results of incorporating spatial information is very accessible and readily applicable.

\begin{figure*}[htbp!]
	\includegraphics[width = 0.95\textwidth]{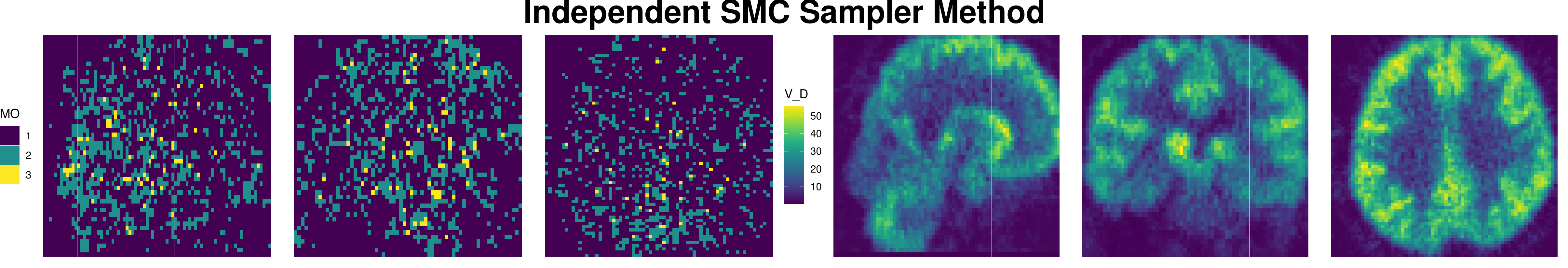} \centering
  \caption{Model order(MO) and volume of distribution(V\_D) parametric image of measured PET data (sagittal(left), coronal(right) and transverse(middle) cross-sections), using spatially independent SMC sampler($N=300, T=500$) model selection.}
  \label{fig:REAL_PET_SMC}       
\end{figure*}

\begin{figure*}[htbp!]
	\includegraphics[width = 0.95\textwidth]{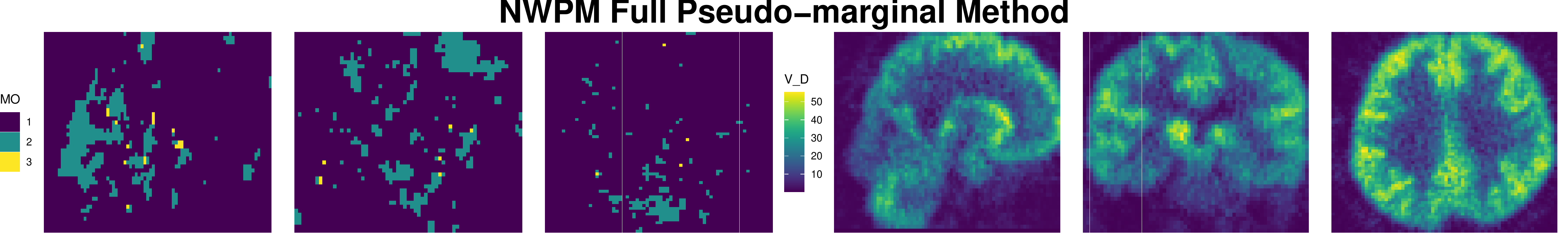} \centering
  \caption{Model order(MO) and volume of distribution(V\_D) parametric image of measured PET data (sagittal(left), coronal(middle) and transverse(right) cross-sections), using the NWPM($n=75$) algorithm for model selection with spatial dependence.
    An SMC sampler($N=300, T=500$) was used for the marginal likelihood estimator.}
  \label{fig:REAL_PET_NWPM}       
\end{figure*}

\begin{figure*}[htbp!]
	\includegraphics[width = 0.95\textwidth]{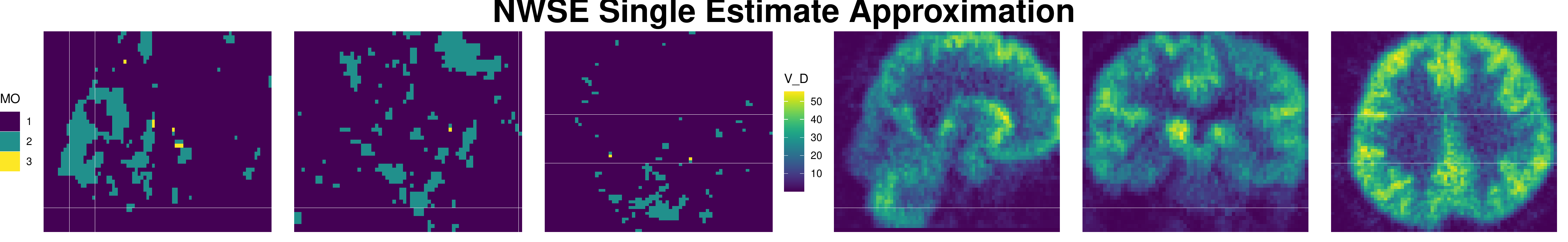} \centering
  \caption{Model order(MO) and volume of distribution(V\_D) parametric image of measured PET data (sagittal(left), coronal(middle) and transverse(right) cross-sections), using the NWSE approximation method($n=500$) for model selection with spatial dependence.
    An SMC sampler($N=400, T=600$) was used for the marginal likelihood estimator.}
  \label{fig:REAL_PET_NWPM_SE}       
\end{figure*}

\begin{figure*}[htbp!]
	\includegraphics[width = 0.95\textwidth]{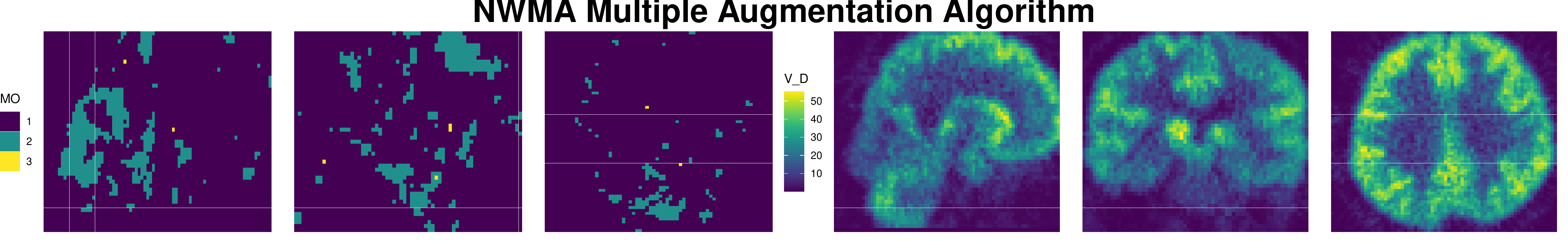} \centering
  \caption{Model order(MO) and volume of distribution(V\_D) parametric image of measured PET data (sagittal(left), coronal(middle) and transverse(right) cross-sections), using the NWMA method($n=500, \kappa=100$) for model selection with spatial dependence.
    An SMC sampler($N=400, T=600$) was used for the marginal likelihood estimator.}
  \label{fig:REAL_PET_NWPM_MA}       
\end{figure*}

\begin{figure*}
	\includegraphics[width = 0.6\textwidth]{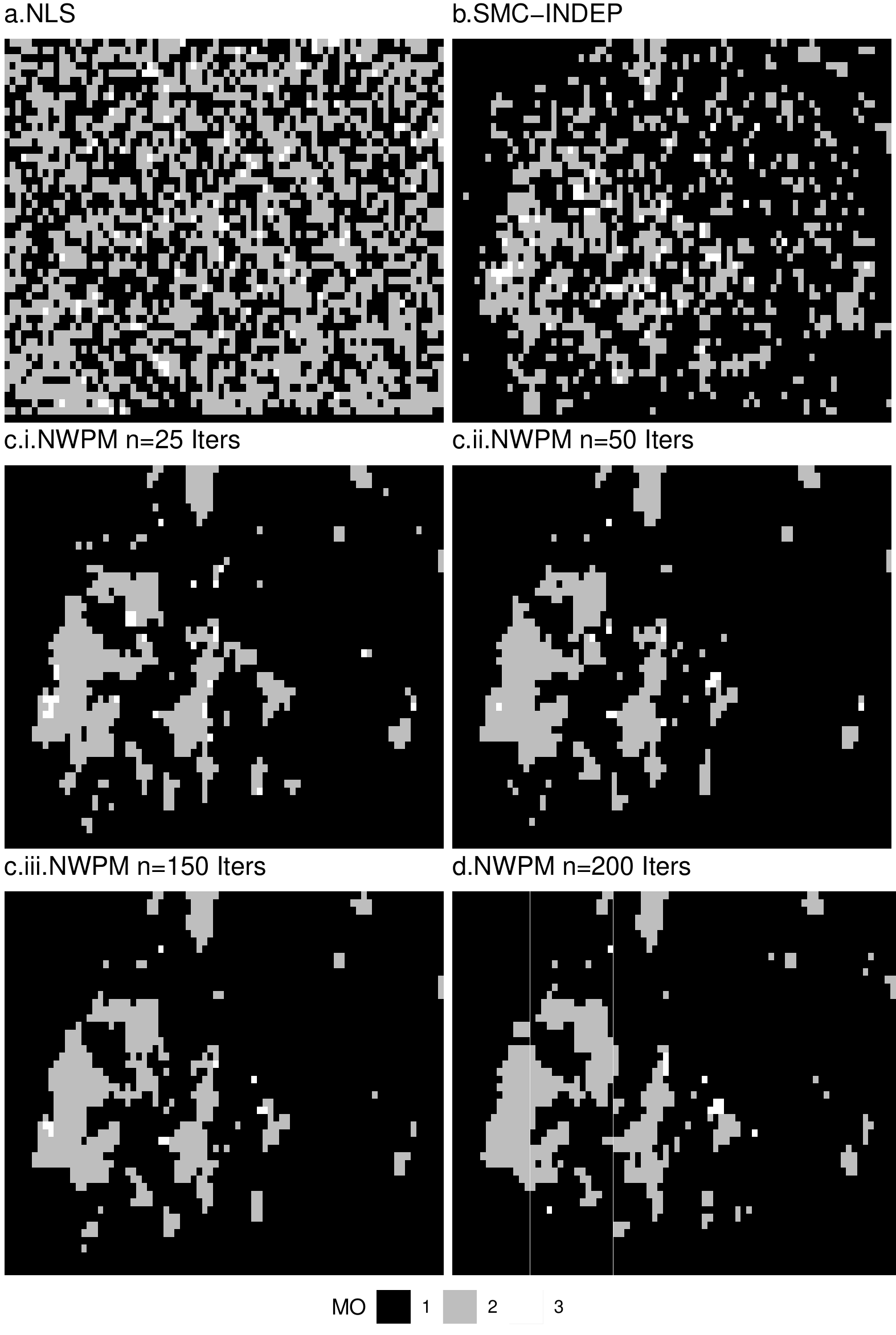} \centering
  \caption{Model selection for measured PET images, using: (a.) spatially independent NLS and (b.) SMC; (d.) NWPM with spatial dependence incorporated using Potts model. Part (c.) shows the progress of the NWPM chain. Model orders (MO) shown.}
  \label{fig:REALPET_ALGO_COMP}       
\end{figure*}

\paragraph{Acknowledgments}
This work was supported by funding from the Engineering and Physical Sciences Research Council under grants EP/\linebreak[1]R034710/1 and EP/\linebreak[1]T004134/1, the Lloyd’s Register Foundation Programme on Data-centric Engineering at the Alan Turing Institute and EPSRC Studentship Project: 1943006.

We thank Prof. Alexander Hammers and Prof. Sir John Aston for providing us with the PET data set used in this study.

\bibliographystyle{unsrtnat}
\bibliography{references}  

\begin{appendices}

\section{PET Model}\label{App:Model_Equations}


The posterior distribution for the PET compartmental models with normally-distributed errors is as follows:
Let 
$$\iota_{i}(\phi,\vartheta) := \frac{C_{\mathrm{T}}(t_i;\phi,\vartheta)}{t_{i}-t_{i-1}},$$
the likelihood for data $y=(y_1, \dots, y_k)$ can be written:
$$\prod_{i=1}^{k}\sqrt{\frac{\lambda}{2\pi \iota_{i}(\phi,\vartheta)}}\exp\left\{ -\frac{\lambda}{2\iota_{i}(\phi,\vartheta)} (y_i-C_{\mathrm{T}}(t_i;\phi,\vartheta))^2 \right\}.
 $$ 
The prior (joint) distributions over $\phi,$ $\vartheta$ and the precision parameter $\lambda = \frac{1}{\sigma_2}$ is given by:
 $$\lambda^{\alpha-1}e^{-\beta\lambda}\prod_{j=1}^{M}I_{[\phi^{a}_j,\phi^{b}_j]} I_{[\vartheta^{a}_j,\vartheta^{b}_j]}$$
 Here $\alpha = \beta = 10^{-3},$ the parameters of the prior distribution of $\lambda$. And $\phi_j^a$ and $\phi_j^a$ are the lower and upper bounds of the truncation interval of parameters $\phi_j$ and corresponding notation is used for $\theta_j$. 

 For the $t-$distributed errors, the observation $y_i$ has a $t$ distribution with location $C_{\mathrm{T}}(t_i)$, scale $\frac{t_i-t_{i-1}}{C_{\mathrm{T}}(t_i)}\tau$, and degrees of freedom $\nu$. The likelihood is, 
 \begin{align*}
   &\prod_{i=1}^{k} \left\{ \frac{\Gamma\left( \frac{\nu+1}{2} \right)}{\Gamma(\frac{\nu}{2})} \left( \frac{\tau}{\iota_i(\phi,\vartheta)\pi \nu} \right)^{\frac{1}{2}} \times \right. \\
   &\left.\left( 1 + \frac{\tau}{\nu \iota_i(\phi,\vartheta)} (y_i - C_{\mathrm{T}}(t_i;\phi,\vartheta))^2\right) ^{-\frac{\nu+1}{2}}\right\}.
 \end{align*}
The prior density is given by:
$$\tau^{\alpha-1}e^{-\beta \tau} \times \frac{1}{\nu^2} \times I_{[a,b]}\left( \frac{1}{\nu} \right) \prod_{i=1}^{M}I_{[\phi_i^a,\phi_i^b]}(\phi_i)I_{[\vartheta_i^{a},\vartheta_i^{b}]}(\vartheta_i),$$

where $\alpha=\beta = 10^{-3}$, the parameters of prior distribution of $\tau;$ $a=0$ and $b=0.5$. 
\clearpage
\onecolumn

\section{MCMC Traces for Studies and Long-run chains}\label{app:Mcmc-traces-long} 
\FloatBarrier
\begin{figure*}[htbp!]
	\includegraphics[width = 0.95 \textwidth]{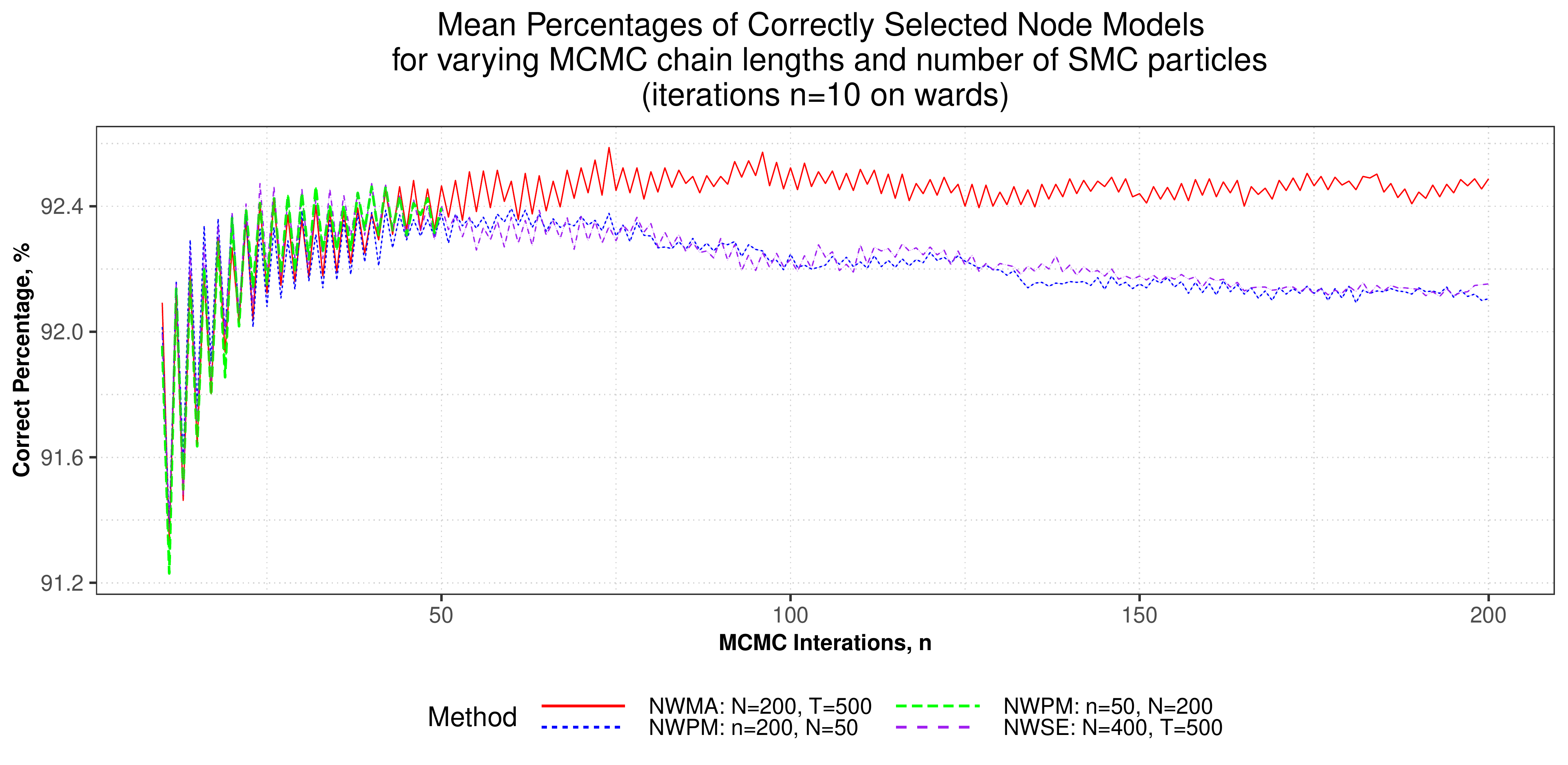} \centering
	\caption{Average percentages(\%) of the whole toy model image ($20 \times 20 = 400$ pixels) where the correct model order was selected at each iteration of the MCMC chain.
    MCMC traces, of each method, for graphical iterations $n=10$ on wards are shown.
  }
  \label{graph:Study2State5_zoom}
\end{figure*}

\begin{figure*}[htbp!]
	\includegraphics[width = 0.95\textwidth]{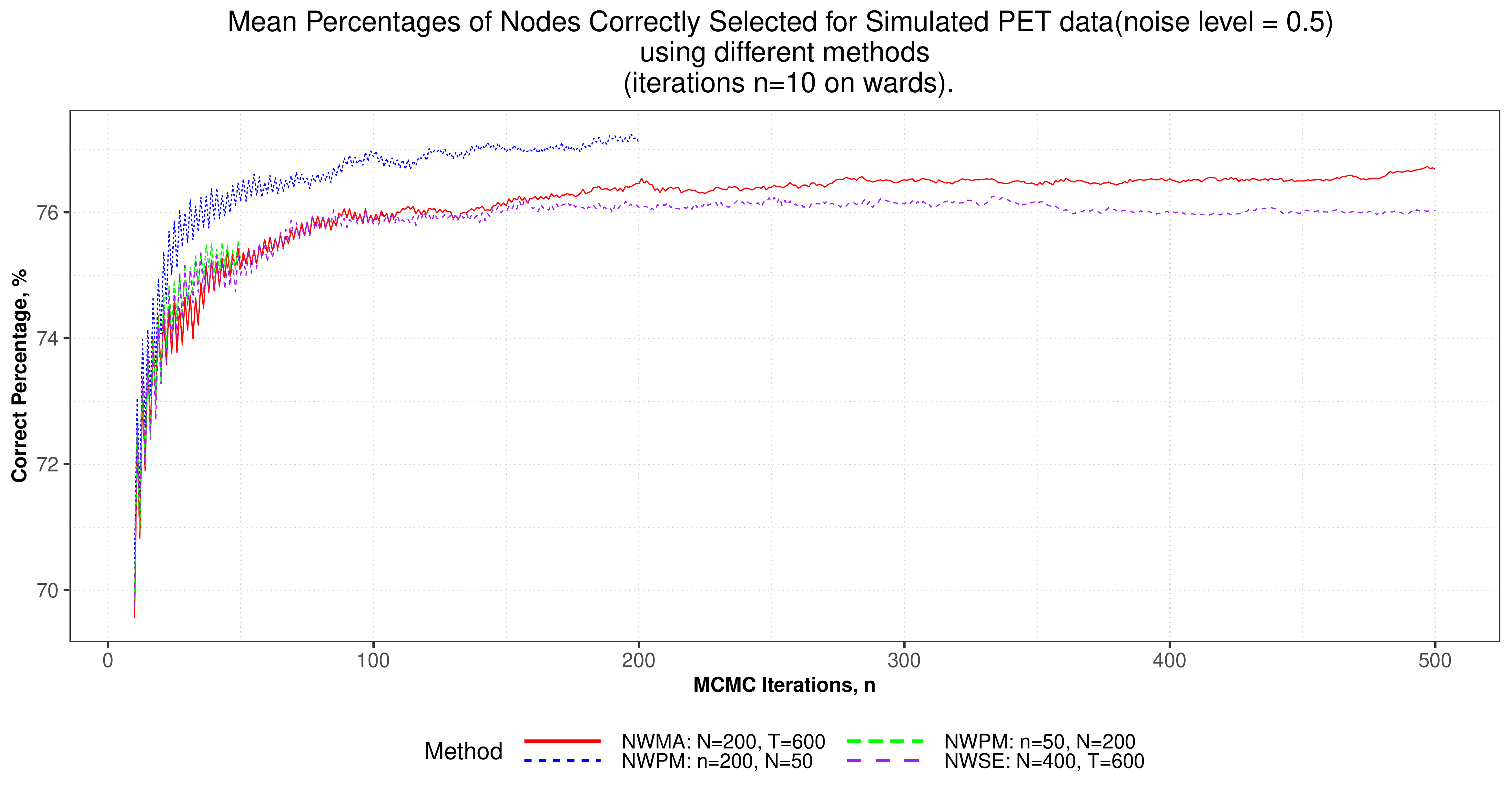} \centering
	\caption{Average percentages(\%) of the whole image ($20 \times 20 = 400$ pixels) where the correct model order was selected at each iteration of the pseudo-marginal MH chain.
    MCMC traces, of each method, for graphical iterations $n=10$ and on wards are shown.
  }
  \label{graph:PETSimStudy_zoom}
\end{figure*}

\begin{figure}[htbp!]
	\includegraphics[width = \textwidth]{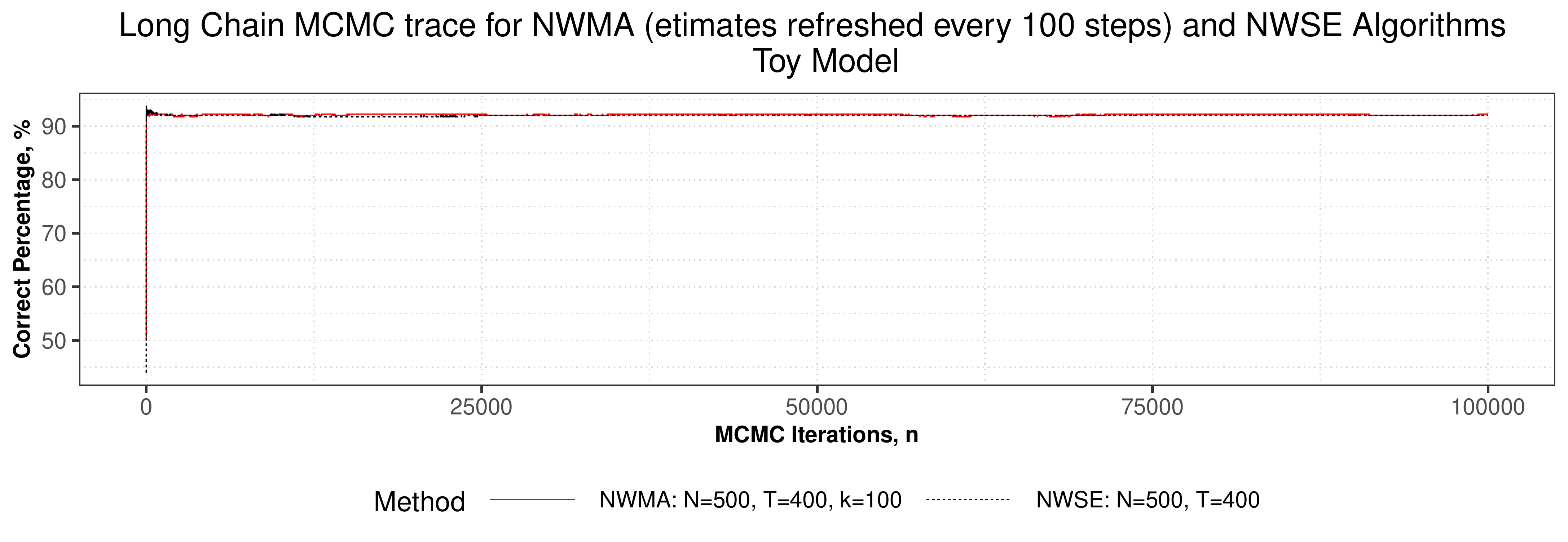} \centering
	\caption{MCMC traces of the percentage of correctly selected nodes, for the toy model. The NWSE and NWMA (refreshing marginal estimates at every $\kappa=100$ graphical iterations). The SMC sampler with $N=400, T=500$ was used as the marginal likelihood estimator for both methods.}
  \label{graph:long_run_toy}
\end{figure}


\begin{figure}[htbp!]
	\includegraphics[width = \textwidth]{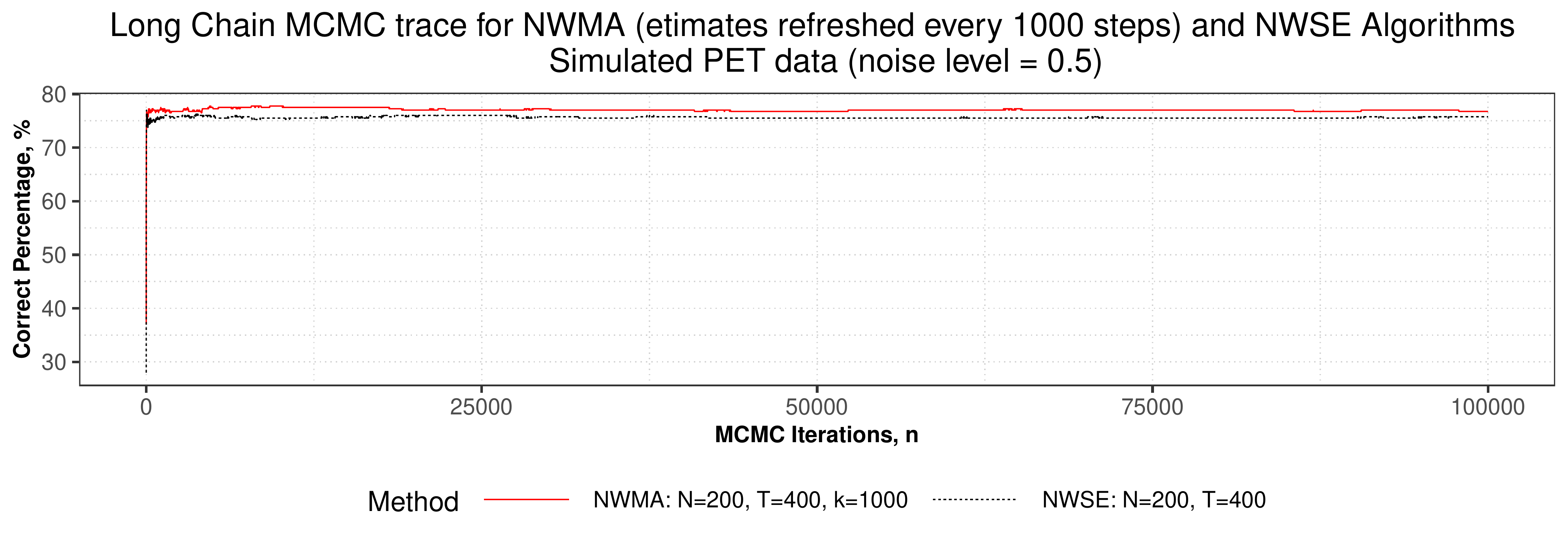} \centering
	\caption{MCMC traces of the percentage of correctly selected nodes, for simulated 2-D PET image, using the NWMA (refreshing marginal estimates at every $\kappa=1000$ graphical iterations) and NWSE methods. For both methods, the SMC sampler with $N=200, T=400$ was used.}
  \label{graph:long_run_PET}
\end{figure}

\begin{figure}[htbp!]
	\includegraphics[width = \textwidth]{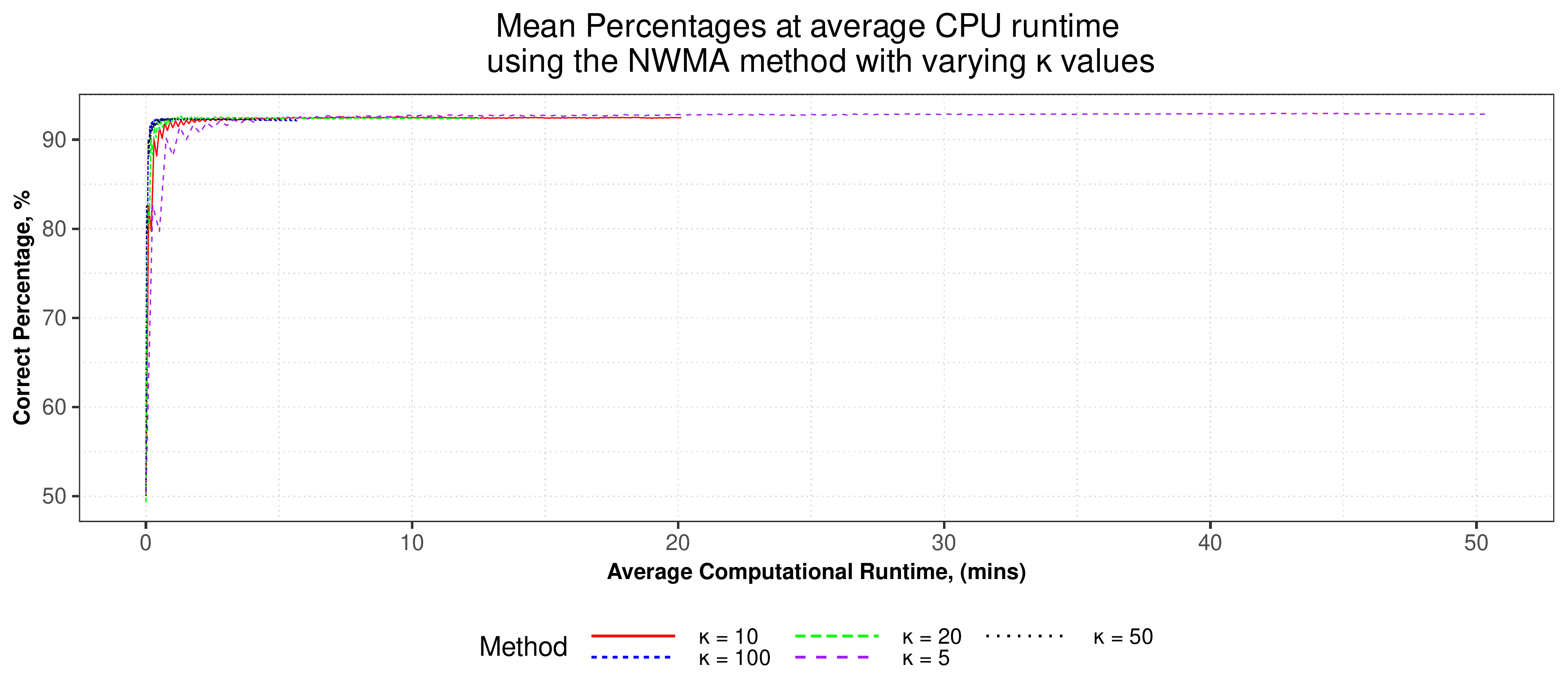} \centering
	\caption{Mean computational runtime (averaged over each graphical iteration), when analysing the Toy Model $20\times 20$ image. The NWMA method (with parameters $n= 200, N=200, T=500$) was used , with $\kappa$ varying from $\kappa=5$ and $\kappa=100$.}
  \label{graph:NWMA_kappa_comp}
\end{figure}

\begin{figure}[htbp!]
	\includegraphics[width = 0.95\textwidth]{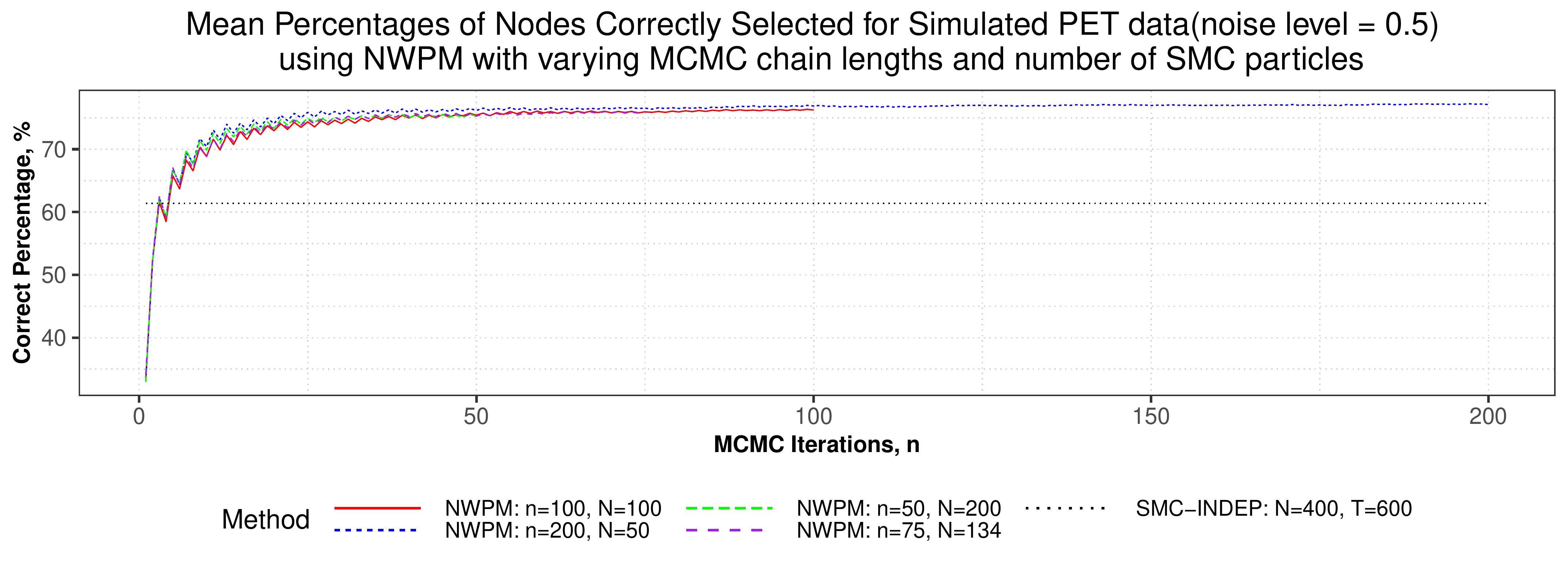} \centering
	\caption{Average percentages(\%) of the whole image where the correct model order was selected at each iteration of MH chain, using the NWPM algorithm for varying Markov Chain length, $n,$ and number of particles, $N,$ in the SMC sampler. The dashed line shows the average percentage when using spatially independent SMC($N=400,T=600$) model selection.}
  \label{graph:PETSimStudyInterRange}
\end{figure}
\clearpage

\section{Model Order outputs for Toy Model}
\begin{figure}[htbp!]
	\includegraphics[width = 0.95\textwidth]{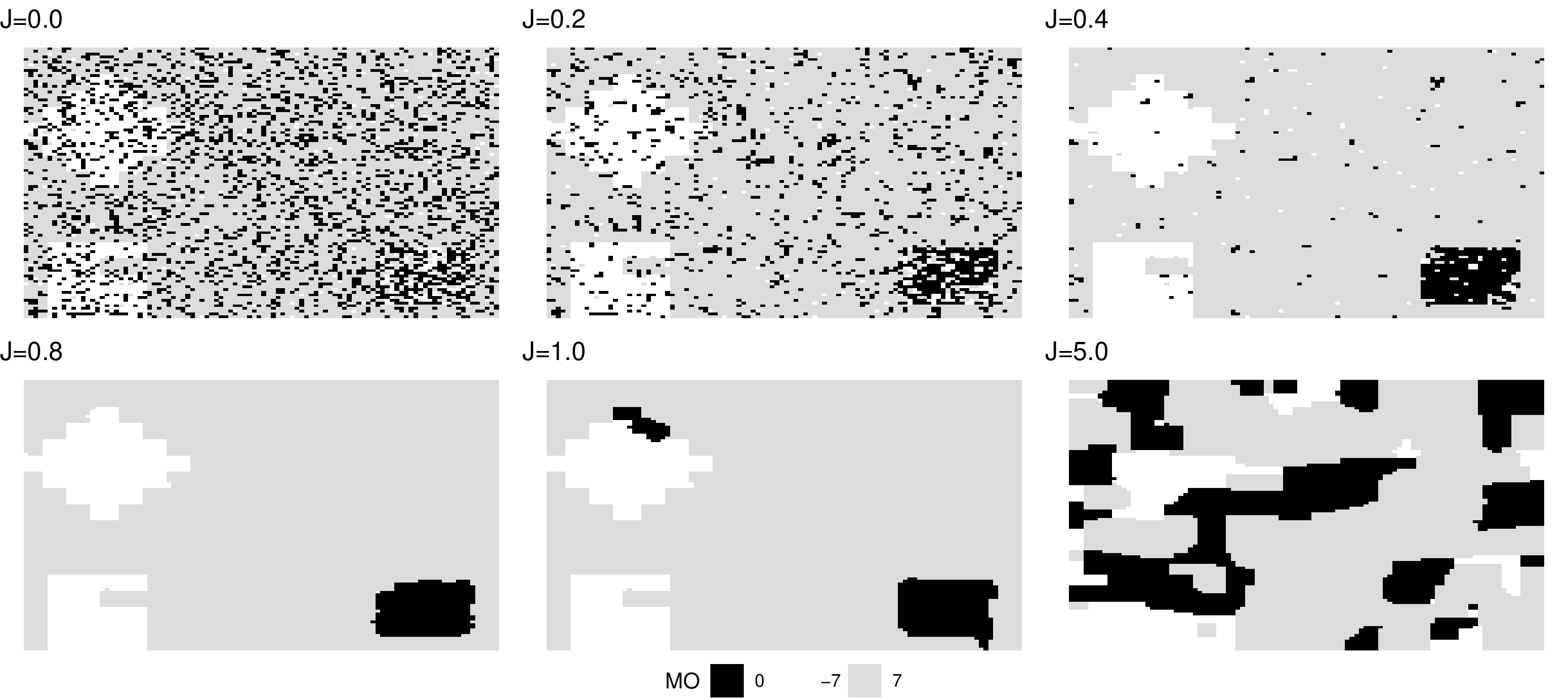} \centering
	\caption{Model order (MO) outputs for the 100$\times$100 Toy Model simulated image, for varying coupling constant values $J$.}
  \label{fig:thre_model_toy}
\end{figure}

\FloatBarrier
\end{appendices}






\end{document}